\documentclass[11pt,a4paper]{article}
\pdfoutput=1

\usepackage{fullpage}
\usepackage[T1]{fontenc}
\usepackage{authblk}
\usepackage{thmtools,thm-restate}
\usepackage{soul}
\usepackage{xcolor}
\usepackage{amsmath,amssymb,amsthm}
\usepackage{caption,subcaption}
\newcommand{\psubref}[1]{(\subref{#1})}
\usepackage{tikz}
\usetikzlibrary{calc} % + notation on coordinates
\usetikzlibrary{decorations.pathreplacing,calligraphy} % curly brakets

\newcommand*\samethanks[1][\value{footnote}]{\footnotemark[#1]}

\newcommand{\ov}{\mathrm{ov}}
\newcommand{\occ}[0]{\mathrm{occ}}
\newcommand{\rev}[1]{\mathrm{rev}(#1)}
\newcommand{\str}[1]{\overline{#1}}
\newcommand{\Tpre}{$T_{pre}$}
\newcommand{\Tsuf}{$T_{suf}$}
\newcommand{\lsplits}{\mathrm{Splits}}
\newcommand{\rsplits}{\mathrm{Splits}_\mathrm{rev}}
\newcommand{\splits}{\mathrm{Splits}'}

\newcommand{\anc}{\mathrm{anc}}
\newcommand{\off}{\mathrm{off}}
\newcommand{\nextnode}{\mathrm{next}}

\newcommand{\lab}{\lambda}
\newcommand{\head}{\textsf{head}}
\newcommand{\tail}{\textsf{tail}}

\newcommand{\Pset}{\mathcal{P}}

\newtheorem{theorem}{Theorem}[section]
\newtheorem{lemma}[theorem]{Lemma}
\newtheorem{corollary}[theorem]{Corollary}
\newtheorem{definition}[theorem]{Definition}
\newtheorem{observation}[theorem]{Observation}

\newtheorem{fact}[theorem]{Fact}
\newtheorem{claim}[theorem]{Claim}

\definecolor{color1}{RGB}{95, 150, 237}
\definecolor{color2}{RGB}{242, 163, 65}

\title{Compressed Indexing for Consecutive Occurrences\footnote{This is a full version of a paper accepted to CPM 2023.}} 
\date{\empty}

\author[1]{Pawe\l{} Gawrychowski\thanks{Partially supported by the Bekker programme of the Polish National Agency for Academic Exchange (PPN/BEK/2020/1/00444) and the grant ANR-20-CE48-0001 from the French National Research Agency (ANR).}}
\author[2]{Garance Gourdel\thanks{Partially supported by the grant ANR-20-CE48-0001 from the French National Research Agency (ANR).}}
\author[3]{Tatiana Starikovskaya\samethanks}
\author[4]{Teresa Anna Steiner\thanks{Supported by a research grant (VIL51463) from VILLUM FONDEN.}}

\affil[1]{Institute of Computer Science, University of Wrocław \\ \texttt{gawry1@gmail.com}}
\affil[2]{DI/ENS, PSL Research University, IRISA Inria Rennes \\ \texttt{garance.gourdel@irisa.fr}}
\affil[3]{DI/ENS, PSL Research University \\ \texttt{tat.starikovskaya@gmail.com}}
\affil[4]{DTU Compute \\ \texttt{terst@dtu.dk}}

\begin{document}
\begin{titlepage}
\thispagestyle{empty} \maketitle

\begin{abstract}
The fundamental question considered in algorithms on strings is that of indexing, that is, preprocessing
a given string for specific queries. By now we have a number of efficient solutions for this problem when
the queries ask for an exact occurrence of a given pattern~$P$. However, practical applications motivate
the necessity of considering more complex queries, for example concerning near occurrences of two patterns. Recently, Bille et al. [CPM 2021] introduced a variant of such queries, called gapped consecutive occurrences,
in which a query consists of two patterns $P_{1}$ and $P_{2}$ and a range $[a,b]$, and one must find all consecutive occurrences $(q_1,q_2)$ of $P_{1}$
and $P_{2}$ such that $q_2-q_1 \in [a,b]$. By their results, we cannot hope for a very efficient indexing structure
for such queries, even if $a=0$ is fixed (although at the same time they provided a non-trivial upper bound). Motivated by
this, we focus on a text given as a straight-line program (SLP) and design an index taking space polynomial in the size of the grammar that answers such queries in time optimal up to polylog factors.
\end{abstract}

\end{titlepage}

\setcounter{page}{1}
\section{Introduction}
In the indexing problem, the goal is to preprocess a string for locating occurrences of a given pattern.
For a string of length $N$, structures such as the suffix tree~\cite{Weiner73} or the suffix array~\cite{ManberM93},
use space linear in $N$ and allow for answering such queries in time linear in the length of the pattern $m$. By now,
we have multiple space- and time-efficient solutions for this problem (both in theory and in practice).
We refer the reader to the excellent survey by Lewenstein~\cite{Lewenstein13} that provides an overview
of some of the approaches and some of its extensions, highlighting its connection to orthogonal range searching.

However, from the point of view of possible applications, it is desirable to allow for more general queries
than just locating an exact match of a given pattern in the preprocessed text, while keeping the time sublinear
in the length of the preprocessed string. A very general query is locating a substring matching a regular
expression.
Very recently, Gibney and Thankachan~\cite{GibneyT21} showed that if the Online Matrix-Vector multiplication conjecture holds, 
even with a polynomial preprocessing time we cannot answer regular expression query in sublinear time.
A more reasonable and yet interesting query could concern occurrences of two given patterns that are closest to each other, or just close enough.

Preprocessing a string for queries concerning two patterns has been first studied in the context of document
retrieval, where the goal is to preprocess a collection of strings. There, in \emph{the two patterns document
retrieval problem} the query consists of two patterns $P_{1}$ and $P_{2}$, and we must report
all documents containing both of them~\cite{Muthukrishnan02}. In \emph{the forbidden pattern query problem}
we must report all documents containing $P_{1}$ but not $P_{2}$~\cite{FischerGKLMSV12}.
For both problems, the asymptotically fastest linear-space solutions need as much as $\Omega(\sqrt{N})$
time to answer a query, where $N$ is the total length of all strings~\cite{HonSTV12,HonSTV10}. That is, the complexity
heavily depends on the length of the strings.
Larsen et al.~\cite{LarsenMNT15} established a connection between Boolean matrix multiplication
and the two problems, thus providing a conditional explanation for the high $\Omega(\sqrt{N})$
query complexity. Later, Kopelowitz et al.~\cite{KopelowitzPP16} provided an even stronger argument
using a connection to the 3SUM problem.
Even more relevant to this paper is the question considered by Kopelowitz and Krauthgamer~\cite{KopelowitzK16},
who asked for preprocessing a string for computing, given two patterns $P_{1}$ and $P_{2}$, their
occurrences that are closest to each other. The main result of their paper is a structure
constructible in $O(N^{1.5}\log^{\epsilon}N)$ time that answers such queries in $O(|P_{1}|+|P_{2}|+\sqrt{N}\log^{\epsilon} N)$, for a string
of length $N$, for any $\epsilon>0$. They also established a connection between Boolean matrix multiplication
and this problem, highlighting a difficulty in removing the $O(\sqrt{N})$ from both the preprocessing 
and query time at the same time.

The focus of this paper is the recently introduced variant of the indexing problem, called \emph{gapped indexing for consecutive occurrences}, in which a query consists of
two patterns $P_{1}$ and $P_{2}$ and a range $[a,b]$, and one must find the pairs of consecutive occurrences of $P_1,P_2$ separated by a distance in the range $[a,b]$. Navarro and Thankanchan~\cite{NavarroT16} showed that for $P_1=P_2$ there is a $O(n \log n)$-space index with optimal query time $O(m+\occ)$, where $m = |P_1|=|P_2|$ and $\occ$ is the number of pairs to report, but in conclusion they noticed that extending their solution to the general case of
two patterns might not be possible. Bille et al.~\cite{cpm/BilleGPS21} provided an evidence of hardness of the general case and established a (conditional) lower bound for gapped indexing for consecutive occurrences,
by connecting its complexity to
that of set intersection. This lower bound suggests that, at least for indexes of size $\tilde O(N)$,
achieving \ul{query time better than $\tilde O(|P_{1}|+|P_{2}|+\sqrt{N})$ would contradict the Set Disjointness conjecture, even if $a=0$ is fixed}. In particular,
obtaining query time depending mostly on the lengths of the patterns (perhaps with some additional logarithms),
arguably the whole point of string indexing, is unlikely in this case.

Motivated by the (conditional) lower bound for gapped indexing for consecutive occurrences, we consider
the compressed version of this problem for query intervals $[0,b]$. For exact pattern matching, there is a long line of research
devoted to designing the so-called compressed indexes, that is, indexing structures with the size being a function of
the length of the compressed representation of the text, see e.g. the entry in the Encyclopedia of Algorithms~\cite{MakinenN16}
or the Encyclopedia of Database Systems~\cite{FerraginaV18}.
This suggests the following research direction: can we design an efficient compressed gapped index for consecutive
occurrences? 

The answer of course depends on the chosen compression method. With a goal to design an index that uses very little space, we focus on the most challenging  setting when the compression is capable of describing a string of exponential length (in the size of its representation). An elegant formalism for such a compression method is that of straight-line programs (SLP), which are context-free grammars describing exactly one string. SLPs are known to capture the popular Lempel--Ziv compression method up to a logarithmic factor~\cite{CharikarLLPPRSS02,Rytter02}, and at the same time provide a more convenient interface, and in particular, allow for random access in $O(\log N)$ time~\cite{random_access_grammar_compress}. 

By now it is known that pattern matching admits efficient indexing in SLP-compressed space. Assuming a string $S$ of length $N$ described by an SLP with $g$ productions, Claude and Navarro~\cite{spire/ClaudeN12a} designed an $O(g)$-space index for $S$ that allows retrieving all occurrences of a pattern of length $m$ in time $O(m^2 \log \log N + \occ \log g)$. 
Recently, several results have improved the query time bound while still using a comparable $O(g\log N)$ amount of space: Claude, Navarro and Pacheco~\cite{CNP2021} showed an index with query time $O((m^2 + \occ)  \log g)$; Christiansen et al.~\cite{talg/ChristiansenEKN21} used strings attractors to further improve the time bound to $O(m + \occ \log^\epsilon N)$; and Díaz-Domínguez et al.~\cite{spire/DNP2021} achieved $O((m \log m +\occ)\log g)$ query time. 

However it is not always the case that a highly compressible string is easier to preprocess.
On the negative side, Abboud et al.~\cite{AbboudBBK17} showed that,
for some problems on compressed strings,
such as computing the LCS, \ul{one cannot completely avoid a high dependency on the length of the uncompressed
string} and that for other problems on compressed strings, such as context-free grammar parsing or RNA folding,
one essentially cannot hope for anything better than just decompressing the string and working with the
uncompressed representation! This is also the case for some problems related to linear algebra~\cite{AbboudBBK20}. Hence, it was not clear to us if one can avoid a high dependency on the length of the uncompressed string
in the gapped indexing for consecutive occurrences problem. 

In this work, we address the lower bound of Bille et al.~\cite{cpm/BilleGPS21} and show that, despite the negative results by Abboud et al.~\cite{AbboudBBK17}, one can circumvent it assuming that the text is very compressible:

\begin{theorem}\label{thm:close_co_occurrences}
For an SLP of size $g$ representing a string $S$ of length $N$, there is an $O(g^5\log^5 N)$-space data structure that maintains the following queries: given two patterns $P_1, P_2$ both of length $O(m)$, and a range $[0,b]$, report all $\occ$ consecutive occurrences of $P_1$ and $P_2$ separated by a distance $d \in [0,b]$. The query time is $O(m\log N + (1+\occ) \cdot \log^4 N \log \log N)$. 
\end{theorem}

While achieving $O(g)$ space and $O(m+\occ)$ query time would contradict the Set Disjointness conjecture by the reduction of Bille et al.~\cite{cpm/BilleGPS21}, one might wonder if the space can be improved without increasing the query time and what is the true complexity of the problem when $a$ is not fixed (recall that $[a,b]$ is the range limiting the distance between co-occurrences to report). While we leave improvement on space and the general case as an interesting open question, we show that in the simpler case $a = 0, b = N$ (i.e. when there is no bound on the distance between the starting positions of $P_1$ and $P_2$), our techniques do allow for $O(g^2\log^4 N)$ space complexity, see Corollary~\ref{cor:all}\footnote{Note that the conditional lower bound of Bille et al.~\cite{cpm/BilleGPS21} does not hold for this simpler case.}.

Throughout the paper we assume a unit-cost RAM model of computation with word size $\Theta(\log N)$. All space
complexities refer to the number of words used by a data structure.

\section{Preliminaries}
\label{sec:prelim}
A \emph{string} $S$ of length $|S| = N$ is a sequence $S[0]S[1]\dots S[N-1]$ of characters from an alphabet~$\Sigma$. We denote the \emph{reverse} $S[N-1] S[N-2] \ldots S[0]$ of $S$ by $\rev{S}$. We define $S[i \dots j]$ to be equal to $S[i] \dots S[j]$ which we call a \emph{substring} of $S$ if $i \le j$ and to the empty string otherwise. We also use notations $S[i \dots j)$ and $S(i\dots j]$ which naturally stand for $S[i] \dots S[j-1]$ and $S[i+1] \dots S[j]$, respectively. 
We call a substring $S[0 \dots i]$ \emph{a prefix} of $S$ and use a simplified notation $S[\dots i]$, and a substring $S[i \dots N-1]$ \emph{a suffix} of $S$ denoted by $S[i \dots]$. We say that $X$ is a \emph{substring} of $S$ if $X = S[i \dots j]$ for some $0 \le i \le j \le N-1$. The index $i$ is called an \emph{occurrence} of $X$ in $S$. 

An occurrence $q_1$ of $P_1$ and an occurrence $q_2$ of $P_2$ form a  \emph{consecutive occurrence (co-occurrence)} of strings $P_1,P_2$ in a string $S$ if there are no occurrences of $P_1,P_2$ between $q_1$ and $q_2$, formally, there should be no occurrences of $P_1$ in $(q_1,q_2]$ and no occurrences of $P_2$ in $[q_1,q_2)$. For brevity, we say that a co-occurrence is \emph{$b$-close} if $q_2-q_1 \le b$.  
 
An integer $\pi$ is a \emph{period} of a string $S$ of length $N$, if $S[i]=S[i+\pi]$ for all $i=0,\dots, N-1-\pi$. The smallest period of a string $S$ is called \emph{the period} of $S$. We say that $S$ is \emph{periodic} if  the period of $S$ is at most $N/2$. We exploit the well-known corollary of the Fine and Wilf's periodicity lemma~\cite{fine1965uniqueness}:

\begin{corollary}\label{cor:arithmetic_progression}
If there are at least three occurrences of a string $Y$ in a string $X$, where $|X| \le 2|Y|$, then the occurrences of $Y$ in $X$ form an arithmetic progression with a difference equal to the period of $Y$. 
\end{corollary}

\subsection{Grammars}
\begin{definition}[Straight-line program~\cite{tit/KiefferY00}]
A \emph{straight-line program} (SLP) $G$ is a context-free grammar (CFG) consisting of a set of non-terminals, a set of terminals, an initial symbol, and a set of productions, satisfying the following properties:
\begin{itemize}
\item A production consists of a left-hand side and a right-hand side, where the left-hand side is a non-terminal $A$ and the right-hand side is either a sequence $BC$, where $B,C$ are non-terminals, or a terminal;
\item Every non-terminal is on the left-hand side of exactly one production;
\item There exists a linear order $<$ on the non-terminals such that $A < B$ whenever $B$ occurs on the right-hand side of the production associated with $A$.
\end{itemize}
\end{definition}

A \emph{run-length straight-line program} (RLSLP) \cite{mfcs/NishimotoIIBT16} additionally allows productions of form $A\rightarrow B^k$ for positive integers $k$, which correspond to concatenating $k$ copies of $B$. If $A$ is associated with a production $A \rightarrow a$, where $a$ is a terminal, we denote $\head(A) = a$, $\tail(A) = \varepsilon$ (the empty string); if $A$ is associated with a production $A \rightarrow BC$, we denote $\head(A) = B$, $\tail(A) = C$; and finally if $A$ is associated with a production $A \rightarrow B^k$, then $\head(A) = B$, $\tail(A) = B^{k-1}$.

The \emph{expansion} $\str{S}$ of a sequence of terminals and non-terminals $S$ is the string that is obtained by iteratively replacing non-terminals by the right-hand sides in the respective productions, until only terminals remain. We say that $G$ \emph{represents} the expansion of its initial symbol.

\begin{definition}[Parse tree]
 The \emph{parse tree} of a SLP (RLSLP) is a rooted tree defined as follows: 
\begin{itemize}
\item The root is labeled by the initial symbol;
\item Each internal node is labeled by a non-terminal;
\item If $S$ is the expansion of the initial symbol, then the $i$th leaf of the parse tree is labeled by a terminal $S[i]$;
\item A node labeled with a non-terminal $A$ that is associated with a production $A\rightarrow BC$, where $B,C$ are non-terminals, has $2$ children labeled by $B$ and $C$, respectively. If $A$ is associated with a production $A\rightarrow a$, where $a$ is a terminal, then the node has one child labeled by $a$.
\item (RLSLP only) A node labeled with non-terminal $A$ that is associated with a  production $A\rightarrow B^k$, where $B$ is a non-terminal, has $k$ children, each labeled by $B$. 
\end{itemize}
\end{definition}

The \emph{size} of a grammar is its number of productions. The \emph{height} of a grammar is the height of the parse tree. We say that a non-terminal $A$ is an \emph{ancestor} of a non-terminal $B$ if there are nodes $u,v$ of the parse tree labeled with $A, B$ respectively, and $u$ is an ancestor of $v$. For a node $u$ of the parse tree, denote by $\off(u)$ the number of leaves to the left of the subtree rooted at $u$. 

\begin{definition}[Relevant occurrences]
Let $A$ be a non-terminal associated with a production $A\rightarrow \head(A)\tail(A)$. We say that an occurrence $q$ of a string $P$ in $\str{A}$ is \emph{relevant with a split~$s$} if $q = |\str{\head(A)}|-s \le |\str{\head(A)}| \le q+|P|-1$.
\end{definition}

For example, in Fig.~\ref{fig:occurrences} the occurrence $q = 3$ of $P=cab$ is a relevant occurrence in $\str{C}$ with a split~$s=1$ but $\str{A}$ contains no relevant occurrences of $P$.

\begin{restatable}{claim}{primaryocc}
\label{claim:primary_occurrence}
Let $q$ be an occurrence of a string $P$ in a string $S$. Consider the parse tree of an RLSLP representing $S$, and let $w$ be the lowest node containing leaves $S[q], S[q+1], \dots, S[q+|P|-1]$ in its subtree, then either
\begin{enumerate}
\item The label $A$ of $w$ is associated with a production $A \rightarrow BC$, and $q-\off(w)$ is a relevant occurrence in $\str{A}$; or
\item The label $A$ of $w$ is associated with a production $A \rightarrow B^r$ and $q-\off(w)=q'+r' |\str{B}|$ for some $0 \le r' \le r$, where $q'$ is a relevant occurrence of $P$ in $\str{A}$.
\end{enumerate}
\end{restatable}
\begin{proof}
Assume first that $A$ is associated with a production $A \rightarrow BC$. We then have that the subtree rooted at the left child of $w$ (that corresponds to $\str{B}$) does not contain $S[q+|P|-1]$ and the subtree rooted at the right child of $w$ (that corresponds to $\str{C}$) does not contain $S[q]$. As a consequence, $q-\off(w)$ is a relevant occurrence in $\str{A}$. 

Consider now the case where $A$ is associated with a production $A \rightarrow B^r$. The leaves labeled by $S[q]$ and $S[q+|P|-1]$ belong to the subtrees rooted at different children of $A$. If $S[q]$ belongs to the subtree rooted at the $(r'+1)$-th child of $A$, then $q'=q-\off(w)-|\str{B}| \cdot r'$ is a relevant occurrence of $P$ in $\str{A}$. 
\end{proof}

\begin{definition}[Splits]
Consider a non-terminal $A$ of an RLSLP $G$. If it is associated with a production $A \rightarrow BC$, define 
$$\lsplits(A,P) = \rsplits(A,P) = \{s : q \text{ is a relevant occurrence of } P \text{ in } \str{A} \text{ with a split } s\}.$$ 
If $A$ is associated with a rule $A \rightarrow B^k$, define 
\begin{align*}
\lsplits(A,P) &=  \{s : q \text{ is a relevant occurrence of } P \text{ in } \str{A} \text{ with a split } s\};\\
\rsplits(A,P) &=  \{|P|-s : q \text{ is a relevant occurrence of } \rev{P} \text{ in } \rev{\str{A}} \text{ with split } s\}.
\end{align*}
Define $\lsplits(G,P)$ $(\rsplits(G,P))$ to be the union of $\lsplits(A,P)$ $(\rsplits(A,P))$ over all non-terminals $A$ in~$G$, and $\splits(G,P) = \lsplits(G,P) \cup \rsplits(G,P)$. 
\end{definition}

We need the following lemma, which can be derived from Gawrychowski~et~al.~\cite{soda/GawrychowskiKKL18}:

\begin{restatable}{lemma}{locallyconsistent}\label{lm:locally_consistent}
Let $G$ be an SLP of size $g$ representing a string $S$ of length $N$, where $g \le N$. There exists a Las Vegas algorithm that builds a RLSLP $G'$ of size $g' = O(g \log N)$ of height $h = O(\log N)$ representing $S$ in time $O(g \log N)$ with high probability. This RLSLP has the following additional property: For a pattern $P$ of length $m$, we can in $O(m\log N)$ time provide a certificate that $P$ does not occur in $S$, or compute the set $\splits(G',P)$. In the latter case, $|\splits(G',P)| = O(\log N)$. 
\end{restatable}

\subsection{Compact Tries}
\label{sec:compact_tries}
We assume the reader to be familiar with the definition of a compact trie (see e.g.~\cite{Gusfield1997}). Informally, a trie is a tree that represents a lexicographically ordered set of strings. The edges of a trie are labeled with strings. We define the label $\lab(u)$ of a node $u$ to be the concatenation of labels on the path from the root to $u$ and an interval $I(u)$ to be the interval of the set of strings starting with $\lab(u)$. From the implementation point of view, we assume that a node $u$ is specified by the interval $I(u)$. The \emph{locus} of a string $P$ is the minimum depth node $u$ such that $P$ is a prefix of $\lab(u)$. 

The standard tree-based implementation of a trie for a generic set of strings $\mathcal{S}= \{S_1, \ldots, S_k\}$ takes $\Theta\left(\sum_{i=1}^k |S_i|\right)$ space. Given a pattern $P$ of length $m$ and $\tau > 0$
suffixes $Q_1,\dots,Q_{\tau}$ of $P$, the trie allows retrieving the ranges of strings in (the lexicographically-sorted) $\mathcal{S}$ prefixed by
$Q_1,\dots,Q_{\tau}$ in $O(m^2)$ time. However, in this work, we build the tries for very special sets of strings only, which allows for a much more efficient implementation based on the techniques of Christiansen et al.~\cite{talg/ChristiansenEKN21}, the proof is given in Appendix~\ref{app:proofs}:

\begin{restatable}{lemma}{tries}\label{lm:tries}
Given an RLSLP $G$ of size $g$ and height $h$. Assume that every string in a set~$\mathcal{S}$ is either a prefix or a suffix of the expansion of a non-terminal of $G$ or its reverse. The trie for $\mathcal{S}$ 
can be implemented in space $O(|\mathcal{S}|)$ to maintain the following queries in $O(m + \tau \cdot (h + \log m))$ time: Given a pattern $P$ of length $m$ and suffixes $Q_i$ of $P$, $1 \le i \le \tau$, find, for each $i$, the interval of strings in the (lexicographically sorted) $\mathcal{S}$ prefixed by $Q_i$. 
\end{restatable}

\section{Relevant, extremal, and predecessor occurrences in a non-terminal}
\label{sec:occurrences}
In this section, we present a data structure that allows various efficient queries, which we will need to prove Theorem~\ref{thm:close_co_occurrences}.
We also show how it can be leveraged for an index in the simpler case of consecutive occurrences $(a = 0, b = N)$. 
Recall that the text $S$ is a string of length $N$ represented by an SLP $G$ of size $g$. By applying Lemma~\ref{lm:locally_consistent}, we transform $G$ into an RLSLP $G'$ of size $g' = O(g \log N)$ and depth $h = O(\log N)$ representing $S$, which we fix from now on. We start by showing that $G'$ can be processed in small space to allow multiple efficient queries:

\begin{restatable}{theorem}{occurrences}\label{th:occurrences}
There is a $O(g^2\log^4 N)$-space data structure for $G'$ that given a pattern $P$ of length $m$ can preprocess it in $O(m \log N + \log^2 N)$ time to support the following queries for a given non-terminal $A$ of $G'$:
\begin{enumerate}
\item Report the sorted set of relevant occurrences of $P$ in $\str{A}$ in $O(\log N)$ time;
\item Decide whether there is an occurrence of $P$ in $\str{A}$ in $O(\log N \log \log N)$ time;
\item Report the leftmost and the rightmost occurrences of $P$ in $\str{A}$, $\str{\head(A)}$, and $\str{\tail(A)}$ in $O(\log^2 N \log \log N)$ time;
\item Given a position $p$, find the rightmost (leftmost) occurrence $q \le p$ ($q \ge p$) of $P$ in $\str{A}$ in $O(\log^3 N \log \log N)$ time (predecessor/successor). 
\end{enumerate}
\end{restatable}
\noindent Before we proceed to the proof, let us derive a data structure to report all consecutive occurrences (co-occurrences) of a given pair of patterns.

\begin{corollary}\label{cor:all}
For an SLP of size $g$ representing a string $S$ of length $N$, there is an $O(g^2\log^4 N)$-space data structure that supports the following queries: given two patterns $P_1, P_2$ both of length $O(m)$, report all $\occ$ co-occurrences of $P_1$ and $P_2$ in $S$. The query time is $O(m\log N + (1+\occ) \cdot \log^3 N \log \log N)$. 
\end{corollary}
\begin{proof}
We exploit the data structure of Theorem~\ref{th:occurrences} for $G'$. To report all co-occurrences of $P_1,P_2$ in $S$, we preprocess $P_1,P_2$ in $O(m  \log N + \log^2 N)$ time and then proceed as follows. Suppose that we want to find the leftmost co-occurrence of $P_1$ and $P_2$ in the string $S[i\dots]$, where at the beginning $i=0$. We find the leftmost occurrence $q'_1$ of $P_1$ with $q'_1\geq i$ (if it exists) by a successor query on the initial symbol of $G'$ (the expansion of which is the entire string~$S$). Then we find the leftmost occurrence $q_2$ of $P_2$ with $q_2\geq q'_1$ (if it exists) by a successor query and the rightmost occurrence $q_1$ of $P_1$ with $q_1\leq q_2$ by a predecessor query. If either $q'_1$ or $q_2$ do not exist, then there are no more co-occurrences in $S[i\dots]$. Otherwise, clearly, $(q_1,q_2)$ is a co-occurrence, and there can be no other co-occurrences starting in $S[i\dots q_2]$. In this case, we return $(q_1,q_2)$ and set $i=q_2+1$. The running time of the retrieval phase is $O(\log^{3}N\log\log N\cdot (\occ+1))$, since we use at most three successor/predecessor queries to either output a new co-occurrence or decide that there are no more co-occurrences.
\end{proof}

\subsection{Proof of Theorem~\ref{th:occurrences}}
The data structure consists of two compact tries \Tpre\ and \Tsuf\ defined as follows. For each non-terminal $A$, we store $\rev{\str{\head(A)}}$ in \Tpre\ and $\str{\tail(A)}$ in  \Tsuf. We augment \Tpre\ and \Tsuf\ by computing their heavy path decomposition: 

\begin{definition}
The \emph{heavy path} of a trie $T$ is the path that starts at the root of $T$ and at each node $v$ on the path branches to the child with the largest number of leaves in its subtree (\emph{heavy} child), with ties broken arbitrarily. The heavy path decomposition is a set of disjoint paths defined recursively, namely it is defined to be a union of the singleton set containing the heavy path of $T$ and the heavy path decompositions of the subtrees of $T$ that hang off the heavy path. 
\end{definition}

For each non-terminal $A$ of $G'$, a heavy path $h_{pre}$ in \Tpre, and a heavy path $h_{suf}$ in \Tsuf, we construct a multiset of points $\Pset(A,h_{pre},h_{suf})$. For every non-terminal $A'$ and nodes $u \in h_{pre}$, $v \in h_{suf}$ the multiset contains a point $(|\lab(u)|, |\lab(v)|)$ iff $A', u, v$ satisfy the following properties:
\begin{enumerate}
\item $A$ is an ancestor of $A'$;
\item  \label{prop:leaves-below} $I(u)$ contains $\rev{\str{\head(A')}} $ and $I(v)$ contains $\str{\tail(A')} $. 
\item $u,v$ are the lowest nodes in $h_{pre}, h_{suf}$, respectively, satisfying Property~\ref{prop:leaves-below}.
\end{enumerate}
(See Fig.~\ref{fig:occurrences}.) The set 
$P(A,h_{pre},h_{suf})$ is stored in a two-sided 2D orthogonal range emptiness data structure~\cite{Lewenstein13,journals/talg/Chan13} which occupies  $O(|\Pset(A,h_{pre},h_{suf})|)$ space. 
 Given a 2D range of the form $[\alpha,\infty]\times[\beta,\infty]$, it allows to decide whether the range contains a point in $\Pset(A,h_{pre},h_{suf})$ in $O(\log\log N)$ time.

\begin{figure}
\centering
\captionsetup[subfigure]{justification=centering}
\begin{subfigure}[b]{0.4\textwidth}
\centering
\includegraphics[width=\textwidth]{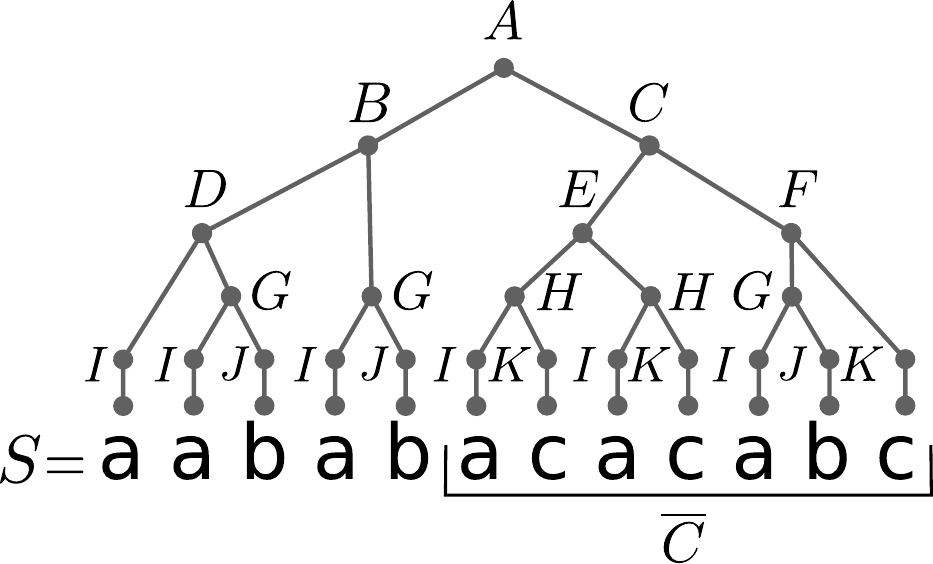}
\caption{Parse tree of $G'$.}
\end{subfigure}
\hfill
\begin{subfigure}[b]{0.58\textwidth}
\centering
\includegraphics[width=\textwidth]{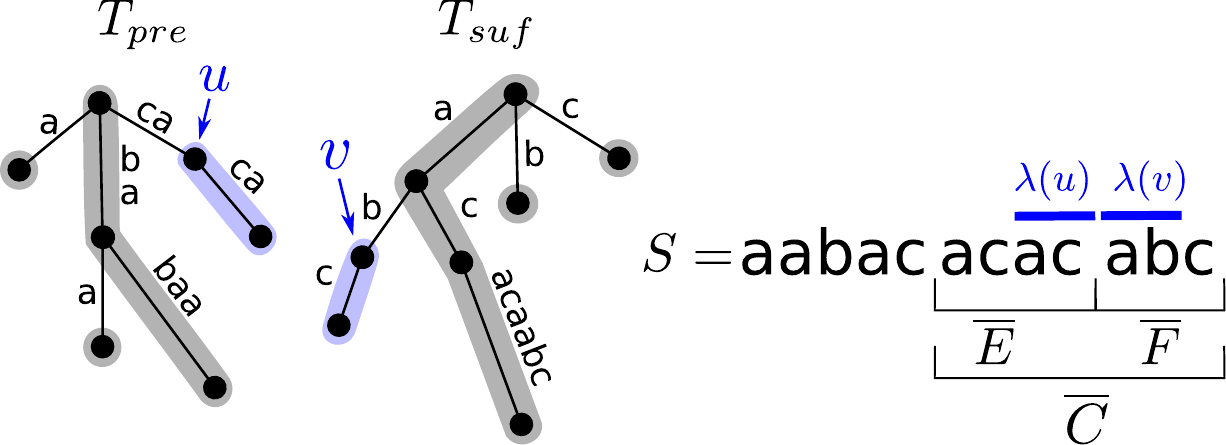}
\caption{Searching for \texttt{cab} with a split $s=1$.}
\end{subfigure}
\caption{A string $S=\mathrm{aababacacabc}$ is generated by an SLP $G'$. Nodes $u$ and $v$ are the loci of \texttt{c} and \texttt{ab} in \Tpre\  and \Tsuf\, respectively. The heavy paths $h_{pre}$ in \Tpre\ and $h_{suf}$ in \Tsuf\ are shown in blue. We have $(2,2) \in \Pset(A,h_{pre},h_{suf})$ corresponding to $C,u,v$.}
\label{fig:occurrences}
\end{figure}

\begin{claim}\label{claim:space-all}
The data structure occupies $O(g^2 \log^4 N)$ space.
\end{claim}
\begin{proof}
Each non-terminal $A'$ has at most $g'$ distinct ancestors and each root-to-leaf path in \Tpre\ or \Tsuf\ crosses $O(\log g')$ heavy paths (as each time we switch heavy paths, the number of leaves in the subtree of the current node decreases by at least a factor of two). As a corollary, each non-terminal creates $O(g' \log^2 g') = O(g \log^3 N)$ points across all orthogonal range emptiness data structures. 
\end{proof}

When we receive a pattern $P$, we compute $\splits(G',P)$ via Lemma~\ref{lm:locally_consistent} in $O(m \log N)$ time or provide a certificate that $P$ does not occur in $S$, in which case there are no occurrences of $P$ in the expansions of the non-terminals of $G'$. Recall that $|\splits(G',P)| \in O(\log N)$. We then sort $\splits(G',P)$ in $O(\log^2 N)$ time (a technicality which will allow us reporting relevant occurrences sorted without time overhead). Finally, we compute, for each $s \in \splits(G', P)$, the interval of strings in \Tpre\ prefixed by $\rev{P[\dots s]}$ (which is the interval $I(u)$ for the locus $u$ of $\rev{P[\dots s]}$ in \Tpre) and the interval of strings in \Tsuf\ prefixed by $P(s \dots ]$ (which is the interval $I(u)$ for the locus $u$ of $P(s \dots]$ in \Tsuf). By Lemma~\ref{lm:tries}, with $\tau=|\splits(G',P)|=O(\log N)$ and $h=O(\log N)$, this step takes $O(m+\log^2 N)$ time.

Reporting relevant occurrences is easy: by definition, each relevant occurrence $q$ of $P$ in $\str{A}$ is equal to $|\str{\head(A)}|-s$ for some  $s \in \splits(G',P)$ such that $\rev{P[\dots s]}$ is a prefix of $\rev{\str{\head(A)}}$ and $P(s\dots]$ is a prefix of $\str{\tail(A)}$. As we already know the intervals of the strings in \Tsuf\ and \Tpre\ starting with $\rev{P[\dots s]}$ and $P(s\dots]$, respectively, both conditions can be checked in constant time per split, or in $O(|\splits(G',P)|) = O(\log N)$ time overall. Note that since $\splits(G',P)$ are sorted, the relevant occurrences are reported sorted as well. 

We now explain how to answer emptiness queries on a non-terminal:
\begin{claim}\label{claim:emptiness}
Let $A$ be a non-terminal labeling a node in the parse tree of $G'$. We can decide whether $\str{A}$ contains an occurrence of $P$ in $O(\log N\log \log N)$ time. 
\end{claim} 
\begin{proof}
Below we show that $P$ occurs in $\str{A}$ iff there exists a split $s \in \splits(G',P)$ such that for $u$ being the locus of $\rev{P[\dots s]}$ in \Tpre\ and $v$ the locus of $P(s \dots]$ in \Tsuf , for $h_{pre}$ the heavy path containing $u$ in \Tpre and $h_{suf}$ the heavy path containing $v$ in \Tsuf , the rectangle $[|\lab(u)|,+\infty] \times [|\lab(v)|,+\infty]$ contains a point from $\Pset(A,h_{pre},h_{suf})$. Before we proceed to the proof, observe that by the bound on $|\splits(G',P)|$ this allows us to decide whether $P$ occurs in $\str{A}$ in $O(\log N)$ range emptiness queries, which results in $O(\log N\log \log N)$ query time. 

Assume that $[|\lab(u)|,+\infty] \times [|\lab(v)|,+\infty]$ contains a point $(x,y) \in \Pset(A,h_{pre},h_{suf})$ corresponding to a non-terminal $A'$. By construction, $A$ is an ancestor of $A'$, the subtree of $u$ contains a leaf corresponding to $\rev{\str{\head(A')}}$ and the subtree of $v$ contains a leaf corresponding to $\str{\tail(A')}$. Consequently, $\str{A'}$ contains an occurrence of $P$, which implies that $\str{A}$ contains an occurrence of $P$. 
To show the reverse direction, let $\ell = \off(u)+1$ and $r = \off(u)+|\str{A}|$, i.e. $S[\ell \dots r] = \str{A}$. The string $\str{A}$ contains an occurrence $\str{A}[q \ldots q+|P|)$ of $P$ iff $S[\ell+q \ldots \ell+q+|P|)$ is an occurrence of $P$ in $S$. From Claim~\ref{claim:primary_occurrence} it follows that if $w$ is the lowest node in the parse tree of $G'$ that contains leaves $S[\ell+q], \dots, S[\ell+q+|P|-1]$ in its subtree and $A'$ is its label, then there exists a split $s \in \splits(G',P)$ such that $\rev{P[\dots s]}$ is a prefix of $\rev{\str{\head(A')}}$ and $P(s\dots ]$ of $\str{\tail(A')}$. By definition of $u$ and $v$, the leaf of \Tpre\ labeled with $\rev{\str{\head(A')}}$ belongs to $I(u)$ and the leaf of \Tsuf\ labeled with $\str{\tail(A')}$ belongs to $I(v)$. Let $h_{pre}$ ($h_{suf}$) be the heavy path in \Tpre (\Tsuf) containing $u$ ($v$) and $(x,y)$ be the point in $\Pset(A,h_{pre},h_{suf})$ created for $A'$. As $|\lab(u)| \le x$ and $|\lab(v)| \le y$, the rectangle $[|\lab(u)|,+\infty] \times [|\lab(v)|,+\infty]$ is not empty.  
\end{proof}

It remains to explain how to retrieve the leftmost/rightmost occurrences in a non-terminal, as well as to answer predecessor/successor queries. The main idea for all four types of queries is to start at any node of the parse tree of $G'$ labeled by $A$ and recurse down via emptiness queries and case inspection. Since the length of the expansion decreases each time we recurse from a non-terminal to its child and the height of $G'$ is $h = O(\log N)$, this allows to achieve the desired query time. We provide full details in Appendix~\ref{app:occurrences}. %\qed

\section{Compressed Indexing for Close Co-occurrences}\label{sec:close}
In this section, we show our main result, Theorem~\ref{thm:close_co_occurrences}. Recall that $S$ is a string of length $N$ represented by an SLP $G$ of size $g$. We start by applying Lemma~\ref{lm:locally_consistent} to transform $G$ into an RLSLP $G'$ of size $g' = O(g \log N)$ and height $h = O(\log N)$ representing $S$. 

The query algorithm uses the following strategy: first, it identifies all non-terminals of $G'$ such that their expansion contains a $b$-close relevant co-occurrence, where a relevant co-occurrence is defined similarly to a relevant occurrence: 

\begin{definition}[Relevant co-occurrence]
Let $A$ be a non-terminal of $G'$. We say that a co-occurrence $(q_1,q_2)$ of $P_1, P_2$ in $\str{A}$ is \emph{relevant} if $q_1 \le |\str{\head(A)}| \le q_2 + |P_2|-1$.   
\end{definition}

Second, it retrieves all $b$-close relevant co-occurrences in each of those non-terminals, and finally, reports all $b$-close co-occurrences by traversing the (pruned) parse tree of $G'$, which is possible due to the following claim:

\begin{restatable}{claim}{relevantcoocc}
\label{claim:relevant_cons_occ}
Assume that $P_2$ is not a substring of $P_1$, and let $(q_1,q_2)$ be a co-occurrence of $P_1, P_2$ in a string $S$. In the parse tree of $G'$, there exists a unique node $u$ such that either
\begin{enumerate}
\item Its label $A$ is associated with a production $A \rightarrow BC$, and $(q_1-\off(u),q_2-\off(u))$ is a relevant co-occurrence of $P_1,P_2$ in $\str{A}$;
\item Its label $A$ is associated with a production $A \rightarrow B^k$, $q_1-\off(u)=q_1'+k' |\str{B}|$, $q_2-\off(u)=q_2'+k' |\str{B}|$ for some $0 \le k' \le k$, where $(q_1',q_2')$ is a relevant co-occurrence of $P_1, P_2$ in $\str{A}$. 
\end{enumerate}
\end{restatable}
\begin{proof}
Let $A$ be the label of the lowest node $u$ in the parse tree that contains leaves $S[q_1], S[q_1+1], \ldots, S[q_2+|P_2|-1]$ in its subtree. Because $P_2$ is not a substring of $P_1$, $A$ cannot be associated with a production $A \rightarrow a$. By definition, $S[\off(u)+1]$ is the leftmost leaf in the subtree of this node. 

Assume first that $A$ is associated with a production $A \rightarrow BC$. We then have that the subtree rooted at the left child of $u$ (labeled by $B$) does not contain $S[q_2+|P_2|-1]$ and the subtree rooted at the right child of $u$ (labeled by $C$) does not contain $S[q_1]$. As a consequence, $(q_1-\off(u),q_2-\off(u))$ is a relevant co-occurrence of $P_1,P_2$ in $\str{A}$. 

Consider now the case where $A$ is associated with a production $A \rightarrow B^k$. The leaves labeled by $S[q_1]$ and $S[q_2+|P_2|-1]$ belong to the subtrees rooted at different children of $A$. If $S[q_1]$ belongs to the subtree rooted at the $k'$-th child of $A$, then $(q_1-\off(u)-|\str{B}| \cdot (k'-1),q_2-\off(u)-|\str{B}| \cdot (k'-1))$ is a relevant co-occurrence of $P_1,P_2$ in $\str{A}$. 
\end{proof}

\subsection{Combinatorial observations}
Informally, we define a set of $O(g^2)$ strings and show that for any patterns $P_1,P_2$ there are two strings $S_1,S_2$ in the set with the following property: whenever the expansion of a non-terminal $A$ in $G'$ contains a pair of occurrences $P_1,P_2$ forming a relevant co-occurrence, there are occurrences of $S_1,S_2$ in the proximity. This will allow us to preprocess the non-terminals of $G'$ for occurrences of the strings in the set and use them to detect $b$-close relevant co-occurrences of $P_1,P_2$. 

Consider two tries, \Tpre\ and \Tsuf: For each production of $G'$ of the form $A\rightarrow BC$, we store $\str{C}$ in  \Tsuf\ and $\rev{\str{B}}$ in \Tpre. For each production of the form $A\rightarrow B^k$, we store $\str{B}$, $\str{B^2}$, $\str{B^{k-2}}$, and $\str{B^{k-1}}$ in \Tsuf\ and the reverses of those strings in \Tpre. For $j \in \{1,2\}$ and $s \in \splits(G', P_j)$ define $S_j(s) = \rev{U} V$, where $U$ is the label of the locus of $\rev{P_j[\dots s]}$ in \Tpre\ and $V$ is the label of the locus of $P_j(s \dots]$ in \Tsuf. Let $l_j(s) = |\rev{U}|$ and $\Delta_j(s) = l_j(s)-s$.  

Consider a non-terminal $A$ such that its expansion $\str{A}$ contains a relevant co-occurrence $(q_1,q_2)$ of $P_1,P_2$. 

\begin{claim}\label{claim:q2}
There exists $s \in \splits(G', P_2)$ such that $p_2 = q_2-\Delta_2(s)$ is an occurrence of $S_2(s)$ in $\str{A}$ and $[p_2,p_2+|S_2(s)|) \supseteq [q_2,q_2+|P_2|)$.
\end{claim}
\begin{proof}
Below we show that there exists a descendant $A'$ of $A$ and a split $s \in \splits(G', P_2)$ such that either $\rev{P_2[\dots s]}$ is a prefix of $\rev{\head(A')}$ and $P_2(s \dots ]$ is a prefix of $\str{\tail(A')}$, or $A'$ is associated with a rule $A' \rightarrow (B')^k$, $\rev{P_2[\dots s]}$ is a prefix of $\rev{\str{(B')^2}}$ and $P_2(s\dots ]$ is a prefix of $\str{(B')^{k-2}}$. The claim follows by the definition of \Tpre, \Tsuf, and $S_2(s)$. 

If $q_2$ is relevant in $\str{A}$, there exists a split $s \in \splits(G', P_2)$ such that $\rev{P_2[\dots s]}$ is a prefix of $\rev{\head(A)}$ and $P_2(s \dots ]$ is a prefix of $\str{\tail(A)}$ by definition. If $q_2$ is not relevant, then $q_2 \ge |\str{\head(A)}|$ by the definition of a co-occurrence. By Claim~\ref{claim:primary_occurrence}, there is a descendant $A'$ of $A$ corresponding to a substring $\str{A}[\ell \dots r]$ for which either $(q_2-\ell)$ is relevant (and then we can repeat the argument above), or $A'$ is associated with a rule $A' \rightarrow (B')^k$ and $(q_2-\ell)-k' \cdot |\str{B'}|$ is relevant, for some $0 \le k' \le k$. Consider the latter case. If $A'=A$, then $k'=1$, as otherwise $q_1 < q_2' = q_2-|\str{B'}| < q_2$ is an occurrence of $P_2$ in $\str{A}$ contradicting the definition of a co-occurrence (recall that $(q_1,q_2)$ is a relevant co-occurrence and hence by definition $q_1 < |\str{\head(A)}|$), and therefore $s = |\str{(B')^2}|-q_2+\ell \in \splits(G', P_2)$, $\rev{P_2[ \dots s]}$ is a prefix of $\rev{(B')^2}$ and $P_2(s \dots ]$ is a prefix of $\str{(B')^{k-2}}$. If $A' \neq A$, then we can analogously conclude that $k'=0$, which implies $s = |\str{B'}|-q_2+\ell \in \splits(G', P_2)$, $\rev{P_2[ \dots s]}$ is a prefix of $\rev{B'}$ and $P_2(s \dots ]$ is a prefix of $\str{(B')^{k-1}}$.
\end{proof}

As the definition of a co-occurrence is not symmetric, $q_1$ does not enjoy the same property. However, a similar claim can be shown:

\begin{lemma}\label{lm:q1}
There exists $s \in \splits(G', P_1)$ and an occurrence $p_1$ of $S_1(s)$ in $\str{A}$ such that $[p_1,p_1+|S_1(s)|) \supseteq [q_1,q_1+|P_1|)$ and at least one of the following holds:
\begin{enumerate}
\item $q_1-\Delta_1(s)$ is an occurrence of $S_1(s)$;
\item $q_2$ is a relevant occurrence of $P_2$ in $\str{A}$, the period of $S_1(s)$ equals the period $\pi_1$ of $P_1$, and there exists an integer $k$ such that $p_1 = q_1-\Delta_1(s)-\pi_1 \cdot k$ and $q_2+\pi_1-1 \le p_1 +|S_1(s)|-1 \le q_2+|P_2|-1$.
\end{enumerate}
\end{lemma}
\begin{proof}
If $q_1$ is a relevant occurrence of $P_1$ in $A$ with a split $s \in \splits(G', P_1)$, then $\rev{P_1[\dots s]}$ is a prefix of $\rev{\str{\head(A)}}$ and $P_1(s \dots ]$ is a prefix of $\str{\tail(A)}$ and therefore the first case holds by the definition of \Tpre\ and \Tsuf. 

Otherwise, by Claim~\ref{claim:primary_occurrence}, there is a descendant $A'$ of $\head(A)$ corresponding to a substring $\str{A}[\ell \dots r]$ for which either $(q_1-\ell)$ is relevant (and then we can repeat the argument above), or $A'$ is associated with a rule $A' \rightarrow (B')^k$ and $(q_1-\ell)-k' \cdot |\str{B'}|$, for some $0 \le k' \le k$, is a relevant occurrence of $P_1$ in $\str{A'}$ with a split $s \in \splits(G',P_1)$. Consider the latter case. We must have (1) $q_1+|P_1|-1+|\str{B'}| \ge r$ or (2) $q_1+|\str{B'}|-1 \ge q_2$, because if both inequalities do not hold, then $q_1 < q_1+|\str{B'}| \le q_2$ is an occurrence of $P_1$ in $\str{A}$, which contradicts the definition of a co-occurrence. Additionally, if (1) holds, then by definition there exists a split $s' \in \splits(G', P_1)$ (which might be different from the split $s$ above) such that $\rev{P_1[ \dots s']}$ is a prefix of $\rev{\str{(B')^{r-1}}}$ and $P_1(s' \dots ]$ is a prefix of $\str{B'}$ and we fall into the first case of the lemma. 

From now on, assume that (2) holds and (1) does not. Since $q_1+|\str{B'}| \le r \le |\str{\head(A)|}$ and $(q_1,q_2)$ is a relevant co-occurrence, $q_2$ must be a relevant occurrence of $P_2$ in $\str{A}$. If $|P_1|-s \le |\str{(B')^2}|$, then $\rev{P_1[ \dots s]}$ is a prefix of $\rev{\str{B'}}$ and $P_1(s \dots ]$ is a prefix of $\str{(B')^2}$ and therefore $q_1-\Delta_1(s)$ is an occurrence of $S_1(s)$. Otherwise, by Fine and Wilf's periodicity lemma~\cite{fine1965uniqueness}, the periods of $\str{A'}$, $P_1$, and $S_1(s)$ are equal, since $P_1$ and hence $S_1(s)$ span at least two periods of $\str{A'}$. By periodicity, $S_1(s)$ occurs at positions $q_1-\Delta_1(s)-|\str{B'}| \cdot k$ of $\str{A}$. Let $p_1$ be the leftmost of these positions which satisfies $p_1+|S_1(s)|-1 \ge q_1+|P_1|-1$. This position is well-defined as (1) does not hold, and furthermore $[q_1,q_1+|P_1|) \subseteq [p_1,p_1+|S_1(s)|)$ as $s \le l_1(s)$ and $|S_1(s)|-l_1(s) \ge |P_1|-s$. We have $p_1 = q_1-\Delta_1(s)-\pi_1 \cdot k''$ for some integer $k''$ (as $|\str{B'}|$ is a multiple of $\pi_1$), and 
$q_2 + \pi_1-1 \le q_1+2|\str{B'}|-1 \le q_1+|P_1|-1 \le p_1+|S_1(s)|-1 \le r < q_2+|P_2|-1$, 
where the last inequality holds as $q_2$ is a relevant occurrence in $\str{A}$. The claim of the lemma follows.
\end{proof}

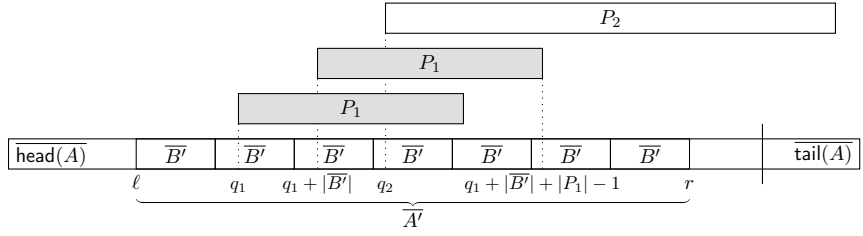
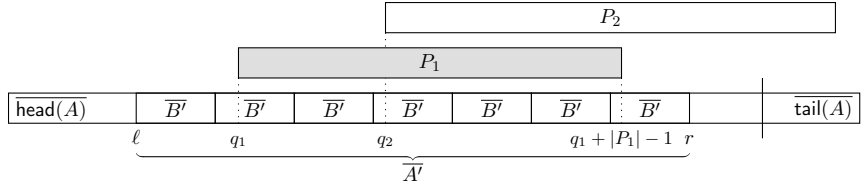
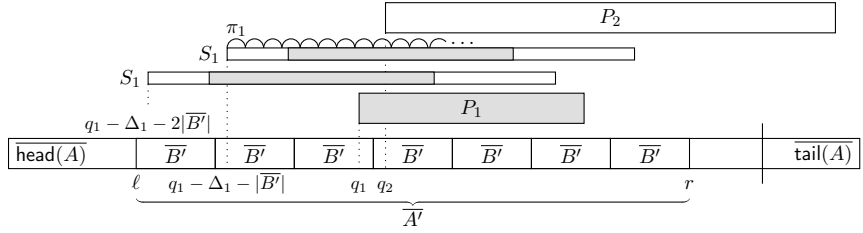
\begin{figure}[h!]
\centering

%------------------------------------------------------------ COMMANDS
\newcommand{\rheight}{0.5}
\newcommand{\Alength}{14}
\newcommand{\Aboundary}{12.4}
\newcommand{\nbB}{6}
\newcommand{\Blength}{1.3}

\newcommand{\nontermA}{
%Non terminal A
\node (ha) at (0,0) {};
\draw (ha) node[right,shift={(0,{0.55*\rheight})}] {$\str{\head(A)}$} rectangle ++ (\Alength,\rheight) node[left, shift={(0,{-0.55*\rheight})}] (ta) {$\str{\tail(A)}$};

% boundary	
\draw (\Aboundary,-0.5*\rheight) -- (\Aboundary,1.5*\rheight);

% B'
\node (l) at (2.1,0) {};
\node[below,text height=0.25cm]  at ($(l)$) {$\ell$};
\foreach \x in {0,...,\nbB}
	\draw ($(l)+ (\x*\Blength,0)$) rectangle ++ (\Blength,\rheight) node[midway] {$\str{B'}$};
\node (r) at ($(l)+(\nbB * \Blength + \Blength,0)$) {};
\node[below,,text height=0.25cm]  at ($(r)$) {$r$};
\draw [decorate,
    decoration = {calligraphic brace,mirror}] ($(l)+(0,-\rheight)$) --  ($(r)+(0,-\rheight)$) node[midway,shift={(0,{-0.7*\rheight})}] {$\str{A'}$};

% spacing
\draw[white] (0,5*\rheight) -- (\Alength,5*\rheight);
}

\newcommand{\Ponelength}{3.7}
\newcommand{\Ptwolength}{7.4}

\newcommand{\Pone}[2][$q_1$]{
\node (q1) at (#2) {};
\draw[fill=gray!25] ($(q1)$) rectangle ++ (\Ponelength,\rheight) node[midway]  {{$P_1$}};
%labels
\path let \p1= (q1) in node (q1label) at (\x1,0) {};
\draw[dotted] ($(q1)$) -- ($(q1label)$) node[below,text height=0.25cm] {\small{#1}};
}

\newcommand{\Sonelength}{6.7}
\newcommand{\Soneheight}{0.2}
\newcommand{\qoneoffset}{0.15*\Sonelength}

\newcommand{\Sone}[2][$p_1$]{
\node (p1) at (#2) {};
\draw[fill=white] ($(p1)$) node[left,shift={(0,{0.5*\Soneheight})}] {$S_1$} rectangle ++ (\Sonelength,\Soneheight);
\draw[fill=gray!25] ($(p1)+(\qoneoffset,0)$) rectangle ++ (\Ponelength,\Soneheight) ;
}

\newcommand{\Ptwo}[1]{
\node (q2) at (#1) {};
\draw[fill=white] (q2) rectangle ++ (\Ptwolength,\rheight) node[midway]  {{$P_2$}};
%labels
\path let \p1= (q2) in node (q2label) at (\x1,0) {};
\draw[dotted] ($(q2)$) -- ($(q2label)$) node[below,text height=0.25cm] {\small{$q_2$}};
}

\newcommand{\pione}{0.15}
\newcommand{\periods}[2]{
\node (s1) at (#1) {};% start of the run
\foreach \x in {1,...,#2} \draw ($(s1)+(\x*\pione*2,0)$) arc(0:180:0.15 and \pione);
\draw ($(s1)+((#2*\pione*2+\pione*2,0.1)$) arc(40:180:0.15 and \pione);
\node at($(s1)+ (\pione, 2*\pione)$) {$\pi_{1}$};
\node at ($(s1)+((#2*\pione*2+\pione*4,0.075)$) {$\dots$};
}

%----------------------------------------------------------------------

\captionsetup[subfigure]{justification=centering}
\begin{subfigure}{\textwidth}
\centering
\begin{tikzpicture}[scale=0.8, every node/.style={scale=0.7}]
\nontermA
\Ptwo{0.5*\Aboundary,4.5*\rheight}
\Pone[$q_1+|\str{B'}|$]{{0.305*\Aboundary+\Blength},3*\rheight}
\path let \p1= (q1) in node (q1end) at ({\x1},0) {};
\draw[dotted] ($(q1)+(\Ponelength,0)$) -- ($(q1end)+(\Ponelength,0)$) node[below,text height=0.25cm] {\small{$q_1+|\str{B'}|+|P_1|-1$}};
\Pone{0.305*\Aboundary,1.5*\rheight}

\end{tikzpicture}
\caption{If neither (1) nor (2), then $(q_1,q_2)$ is not consecutive.}
\label{subfig:neither1nor2}
\end{subfigure}
\begin{subfigure}{\textwidth}
\centering
\begin{tikzpicture}[scale=0.8, every node/.style={scale=0.7}]
\nontermA
\Ptwo{0.5*\Aboundary,3*\rheight}
\renewcommand{\Ponelength}{6.3}
\Pone{0.305*\Aboundary,1.5*\rheight}
\path let \p1= (q1) in node (q1end) at ({\x1},0) {};
\draw[dotted] ($(q1)+(\Ponelength,0)$) -- ($(q1end)+(\Ponelength,0)$) node[below,text height=0.25cm] {\small{$q_1+|P_1|-1$}};
\end{tikzpicture}
\caption{(1) holds and (2) does not.}
\label{subfig:only1holds}
\end{subfigure}
\begin{subfigure}{\textwidth}
\centering
\begin{tikzpicture}[scale=0.8, every node/.style={scale=0.7}]
\nontermA
% spacing
\draw[white] (0,5*\rheight) -- (\Alength,6*\rheight);

\Ptwo{0.5*\Aboundary,4.5*\rheight}

\periods{0.29*\Aboundary,{3.6*\rheight+\Soneheight}}{11};
\Sone{0.29*\Aboundary,3.6*\rheight}
\path let \p1= (p1) in node (p1end) at ({\x1},0) {};
\draw[dotted] ($(p1)$) -- ($(p1end)$) node[below,text height=0.25cm] {\small{$q_1-\Delta_1-|\str{B'}|$}};

\Sone{0.29*\Aboundary-\Blength,2.8*\rheight}
\path let \p1= (p1) in node (p1end) at ({\x1},{2*\rheight}) {};
\draw[dotted] ($(p1)$) -- ($(p1end)$) node[below,text height=0.25cm] {\small{$q_1-\Delta_1-2|\str{B'}|$}};
\Pone{0.465*\Aboundary,1.5*\rheight}
\end{tikzpicture}
\caption{(2) holds, (1) does not, and $|P_1|-s  \geq \str{(B')^2}$.}
\label{subfig:only2holds}
\end{subfigure}

%-----------------------
\caption{Subcases of Lemma~\ref{lm:q1}.}
\label{fig:lemq1}
\end{figure}

We summarize Claim~\ref{claim:q2} and Lemma~\ref{lm:q1}:

\begin{corollary}\label{cor:q1_and_q2}
Let $(q_1,q_2)$ be a co-occurrence of $P_1,P_2$ in the expansion of a non-terminal $A$. There exist splits $s_1 \in \splits(G',P_1), s_2 \in \splits(G',P_2)$ and occurrences $p_1$ of $S_1(s_1)$ and $p_2$ of $S_2(s)$, where $[p_1,p_1+|S_1(s_1)|) \supseteq [q_1,q_1+|P_1|)$ and $[p_2,p_2+|S_2(s_2)|) \supseteq [q_2,q_2+|P_2|)$, such that at least one of the following holds:
\begin{enumerate}
\item \label{it:aperiodic} The occurrence $p_1$ is either relevant or $p_1+|S_1(s_1)|-1 \le |\str{\head(A)}|$. The occurrence $p_2$ is either relevant or $p_2 > |\str{\head(A)}|$. Additionally, $p_1 = q_1 - \Delta_1(s_1)$ and $p_2 = q_2 - \Delta_2(s_2)$. 
\item \label{it:periodic} The occurrence $p_2$ is relevant and $p_1 \le |\str{\head(A)}|$. Additionally, $p_2 = q_2 - \Delta_2(s_2)$, the period of $S_1(s)$ equals the period $\pi_1$ of $P_1$, and there exists an integer $k$ such that $p_1 = q_1-\Delta_1(s_1)-\pi_1 \cdot k$ and $p_2+\pi_1-1 \le p_1 +|S_1(s_1)|-1 \le p_2+|S_2(s_2)|-1$.
\end{enumerate}
\end{corollary}

The reverse observation holds as well:

\begin{observation}\label{obs:rev_q1_and_q2}
If $p_j$ is an occurrence of $S_j(s)$ in $\str{A}$, $j = 1,2$, then $q_j = p_j+\Delta_j(s)$ is an occurrence of $P_j$. Furthermore, if $S_1(s)$ is periodic with period $\pi_1$, then $q_1 + \pi_1 \cdot k$, $0 \le k \le \lfloor (|S_1(s)|-q_1-|P_1|)/\pi_1\rfloor$, are occurrences of $P_1$ in $\str{A}$. 
\end{observation}

Finally, the following trivial observation will be important for upper bounding the time complexity of our query algorithm:

\begin{observation}\label{obs:close}
If a string contains a pair of occurrences $(q_1,q_2)$ of $P_1$ and $P_2$ such that $0 \leq q_2-q_1\leq b$, then it contains a $b$-close co-occurrence of $P_1$ and $P_2$.
\end{observation}

\subsection{Index}
The first part of the index is the data structure of Theorem~\ref{th:occurrences} and the index of Christiansen et al.~\cite{talg/ChristiansenEKN21}:

\begin{fact}[{\cite[Introduction and Theorem 6.12]{talg/ChristiansenEKN21}}]
There is a $O(g \log^2 N)$-space data structure that can find the $\occ$ occurrences of any pattern $P[1\dots m]$ in $S$ in time $O(m + \occ)$.
\end{fact}

The second part of the index are the tries \Tpre\ and \Tsuf, augmented as explained below. Consider a quadruple $(u_1,u_2,v_1,v_2)$, where $u_1$ and $u_2$ are nodes of \Tpre\ and $v_1$ and $v_2$ are nodes of \Tsuf. Let $U_1, U_2, V_1, V_2$ be the labels of $u_1, u_2, v_1, v_2$, respectively. Define $S_1=\rev{U_1}V_1$ and $S_2=\rev{U_2}V_2$, and let $l_1 = |\rev{U_1}|$ and $l_2 = |\rev{U_2}|$.  

First, we store a binary search tree $\mathcal{T}_1(u_1, u_2, v_1, v_2)$ that for each non-terminal $A$ contains at most six integers $d = p_2-p_1$, where $p_1, p_2$ are occurrences of $S_1,S_2$ in $\str{A}$, satisfying at least one of the below:
\begin{enumerate}
\item $p_1$ is the rightmost occurrence of $S_1$ such that $p_1+|S_1|-1 < |\str{\head(A)}|$ and $p_2$ is the leftmost occurrence of $S_2$ such that $p_2 \ge |\str{\head(A)}|$;
\item $p_1$ is a relevant occurrence of $S_1$ with a split $l_1$ and $p_2$ is the leftmost occurrence of $S_2$ such that $p_2 \ge |\str{\head(A)}|$;
\item $p_1$ is a relevant occurrence of $S_1$ with a split $l_1$, $p_2$ is a relevant occurrence of $S_2$ with a split $l_2$; 
\item $p_2$ is a relevant occurrence of $S_2$ with a split $l_2$ and $p_1$ is the rightmost occurrence of $S_1$ such that $p_1+|S_1|-1<p_2$;
\item $p_2$ is a relevant occurrence of $S_2$ with a split $l_2$ and $p_1$ is the leftmost or second leftmost occurrence of $S_1$ in $\str{\head(A)}$ 
such that $p_1 < p_2 \le p_1+|S_1|-1 < p_2+|S_2|-1$. 
\end{enumerate}

Second, we store a list of non-terminals $\mathcal{L}(u_2, v_2)$ such that their expansion contains a relevant occurrence of $S_2$ with a split $l_2$. Additionally, for every $k \in [0, \log N]$, we store, if defined:
\begin{enumerate}
\item The rightmost occurrence $p_1$ of $S_1$ in $S_2$ such that $p_1+(|S_1|-1) \le l_2-2^k$;
\item The leftmost occurrence $p'_1$ of $S_1$ in $S_2$ such that $p'_1 \le l_2-2^k \le p'_1+|S_1|-1$;
\item The rightmost occurrence $p''_1$ of $S_1$ in $S_2$ such that $p''_1 \le l_2-2^k \le p''_1+|S_1|-1$.
\end{enumerate}

Finally, we compute and memorize the period $\pi_1$ of $S_1$. If the period is well-defined (i.e., $S_1$ is periodic), we build a binary search tree $\mathcal{T}_2(u_1, u_2, v_1, v_2)$. Consider a non-terminal~$A$ containing a relevant occurrence $p_2$ of $S_2$ with a split $l_2$.
Let $p_1$ be the leftmost occurrence of $S_1$ such that $p_1 \le  p_2 \le p_1+|S_1|-1 \le p_2 + |S_2|-1$ and $p_1'$ the rightmost. If $p_1$ and $p_1'$ exist ($p_1$ might be equal to $p_1'$) and $p_1'+|S_1|-1 \ge p_2 + \pi_1-1$, we add an integer $(p_1'-p_1)/\pi_1$ to the tree and associate it with $A$. We also memorize a number $\ov(S_1,S_2) = p_2-p_1'$, which does not depend on $A$ by Corollary~\ref{cor:arithmetic_progression} and therefore is well-defined (it corresponds to the longest prefix of $S_2$ periodic with period $\pi_1$).

\begin{claim}
The data structure occupies $O(g^5 \log^5 N)$ space.
\end{claim}
\begin{proof}
The data structure of Theorem~\ref{th:occurrences} occupies $O(g^2\log^4 N)$ space. The index of Christiansen et al. occupies $O(g \log^2 N)$ space. The tries, by Lemma~\ref{lm:tries}, use $O(g') = O(g \log N)$ space. There are $O((g')^4)$ quadruples $(u_1,u_2,v_1,v_2)$ and for each of them the trees take $O(g')$ space. The arrays of occurrences of $S_1$ in $S_2$ use $O(\log N)$ space. Therefore, overall the data structure uses $O(g^5 \log^5 N)$ space.
\end{proof}

\subsection{Query}
Recall that a query consists of two strings $P_1, P_2$ of length at most $m$ each and an integer $b$, and we must find all $b$-close co-occurrences of $P_1,P_2$ in $S$, let $\occ$ be their number. 

We start by checking whether $P_2$ occurs in $P_1$ using a linear-time and constant-space pattern matching algorithm such as~\cite{constantspacepm}. If it is, let $q_2$ be the position of the first occurrence. If $q_2 > b$, then there are no $b$-close co-occurrences of $P_1,P_2$ in $S$. Otherwise, to find all $b$-close co-occurrences of $P_1, P_2$ in $S$ (that \emph{always} consist of an occurrence of $P_1$ in $S$ and the first occurrence of $P_2$ in $P_1$), it suffices to find all occurrences of $P_1$ in $S$, which we do using the index of Christiansen et al.~\cite{talg/ChristiansenEKN21} in time $O(|P_1|+\occ) = O(m+\occ)$. 

From now on, assume that $P_2$ is not a substring of $P_1$. Let $\mathcal{N}$ be the set of all non-terminals in $G'$ such that their expansion contains a relevant $b$-close co-occurrence of $P_1, P_2$. By Claim~\ref{claim:relevant_cons_occ}, $|\mathcal{N}| \le \occ$. 

\begin{lemma}\label{lm:non-term_close_co_occ}
Assume that $P_2$ is not a substring of $P_1$. One can retrieve in $O(m+(1+\occ)\log^3 N)$ time a set $\mathcal{N}' \supset \mathcal{N}$, $|\mathcal{N}'| = O(\occ \log N)$. 
\end{lemma}
\begin{proof}
We start by computing $\splits(G',P_1)$ and $\splits(G',P_2)$ via Lemma~\ref{lm:locally_consistent} in $O((|P_1|+|P_2|) \log N) = O(m \log N)$ time (or providing a certificate that either $P_1$ or $P_2$ does not occur in $S$, in which case there are no co-occurrences of $P_1,P_2$ in $S$ and we are done). Recall that $|\splits(G',P_1)|, |\splits(G',P_2)| \in O(\log N)$. For each fixed pair of splits $s_1 \in \splits (G',P_1)$, $s_2 \in \splits (G',P_2)$ and $j \in \{1,2\}$, we compute the interval of strings in \Tpre\ prefixed by $\rev{P_j[\dots s_j]}$, which corresponds to the locus $u_j$ of $\rev{P_j[\dots s_j]}$ in \Tpre\, and the interval of strings in \Tsuf\ prefixed by $P_j(s_j \dots ]$, which corresponds to the locus $v_j$ of $P_j(s_j \dots ]$ in \Tsuf. Computing the intervals takes $O(m+\log^2 N)$ time for all the splits by Lemma~\ref{lm:tries}. Consider the strings $S_1=\rev{U_1}V_1$ and $S_2=\rev{U_2}V_2$, where $U_1, U_2, V_1, V_2$ are the labels of $u_1, v_1, u_2, v_2$, respectively. Let $l_1 = |\rev{U_1}|$, $\Delta_1 = l_1-s_1$, $l_2 = |\rev{U_2}|$, $\Delta_2 = l_2-s_2$, and $\Delta = \Delta_1-\Delta_2$. 

Consider a relevant co-occurrence $(q_1,q_2)$ of $P_1,P_2$ in the expansion of a non-terminal~$A$. By Corollary~\ref{cor:q1_and_q2}, $q_1,q_2$ imply existence of occurrences $p_1,p_2$ of $S_1,S_2$ such that $[p_1,p_1+|S_1|) \supseteq [q_1,q_1+|P_1|)$ and $[p_2,p_2+|S_2|) \supseteq [q_2,q_2+|P_2|)$. 
Our index must treat both cases of Corollary~\ref{cor:q1_and_q2}. We consider eight subcases defined in Fig.~\ref{fig:q1_and_q2}, which describe all possible locations of $p_1$ and $p_2$.

\begin{figure}[h!]
\centering

%------------------------------------------------------------ COMMANDS
\newcommand{\rheight}{0.5}
\newcommand{\Alength}{5.3}
\newcommand{\Aboundary}{2.4}

\newcommand{\nontermA}{
%Non terminal A
\node (ha) at (0,0) {};
\draw (ha) node[right,shift={(0,{0.55*\rheight})}] {$\str{\head(A)}$} rectangle ++ (\Alength,\rheight) node[left, shift={(0,{-0.55*\rheight})}] (ta) {$\str{\tail(A)}$};

% boundary	
\draw (\Aboundary,-0.25*\rheight) -- (\Aboundary,1.25*\rheight);

% spacing
\draw[white] (0,4*\rheight) -- (\Alength,4*\rheight);
}

\newcommand{\Sonelength}{1.2}
\newcommand{\Ponelength}{0.5}
\newcommand{\Stwolength}{1.8}
\newcommand{\Ptwolength}{0.8}
\newcommand{\qoneoffset}{0.4*\Sonelength}
\newcommand{\qtwooffset}{0.3*\Stwolength}

\newcommand{\Sone}[1]{
\node (p1) at (#1) {};
\node (q1) at ($(p1)+(\qoneoffset,0)$) {};
\draw[fill=white] (p1) node[left,shift={(0,{0.5*\rheight})}]  {{$S_1$}} rectangle ($(p1)+(\Sonelength,\rheight)$);
\draw[fill=gray!50] ($(q1)$) rectangle ++ (\Ponelength,\rheight) node[midway]  {{$P_1$}};
%labels
\path let \p1= (q1) in node (q1label) at (\x1,0) {};
\draw[dotted] ($(q1)$) -- ($(q1label)$) node[below] {{$q_1$}};
\path let \p1= (p1) in node (p1label) at (\x1,0) {};
\draw[dotted] ($(p1)$) -- ($(p1label)$) node[below] {{$p_1$}};
}

\newcommand{\Stwo}[1]{
\node (p2) at (#1) {};
\draw (p2) node[left,shift={(0,{0.5*\rheight})}]  {{$S_2$}} rectangle ($(#1)+(\Stwolength,\rheight)$);
\node (q2) at ($(p2)+(\qtwooffset,0)$) {};
\draw[fill=gray!50] (q2) rectangle ++ (\Ptwolength,\rheight) node[midway]  {{$P_2$}};
%labels
\path let \p1= (q2) in node (q2label) at (\x1,0) {};
\draw[dotted] ($(q2)$) -- ($(q2label)$) node[below] {{$q_2$}};
\path let \p1= (p2) in node (p2label) at (\x1,0) {};
\draw[dotted] ($(p2)$) -- ($(p2label)$) node[below] {{$p_2$}};
}

%----------------------------------------------------------------------

\renewcommand\thesubfigure{1.\arabic{subfigure}}
\captionsetup[subfigure]{justification=centering}
\begin{subfigure}{0.32\textwidth}
\centering
\begin{tikzpicture}[scale=0.8, every node/.style={scale=0.7}]
\renewcommand{\Aboundary}{3}
\nontermA
\Stwo{1.1*\Aboundary,1.5*\rheight}
\Sone{0.5*\Aboundary,1.5*\rheight}
\node at (2.1,3*\rheight) {\small{rightmost}};
\node at (4.2,3*\rheight) {\small{leftmost}};

\end{tikzpicture}
\caption{}
\label{subfig:both_irrelevant}
\end{subfigure}
\begin{subfigure}{0.32\textwidth}
\centering
\begin{tikzpicture}[scale=0.8, every node/.style={scale=0.7}]
\nontermA
\Stwo{1.45*\Aboundary,1.5*\rheight}
\Sone{0.6*\Aboundary,1.5*\rheight}
\node at (4.4,3*\rheight) {\small{leftmost}};

\end{tikzpicture}
\caption{}
\label{subfig:q1_relevant}
\end{subfigure}
\begin{subfigure}{0.32\textwidth}
\centering
\begin{tikzpicture}[scale=0.8, every node/.style={scale=0.7}]
\nontermA
\Stwo{0.9*\Aboundary,3*\rheight}
\Sone{0.52*\Aboundary,1.5*\rheight}

\end{tikzpicture}
\caption{}
\label{subfig:both_relevant}
\end{subfigure}

%-----------------------
\vspace{0.5cm}

\begin{subfigure}{0.32\textwidth}
\centering
\begin{tikzpicture}[scale=0.8, every node/.style={scale=0.7}]
\renewcommand{\Aboundary}{3.5}
\nontermA
\Stwo{0.95*\Aboundary,1.5*\rheight}
\Sone{0.4*\Aboundary,1.5*\rheight}
\node at (2,3*\rheight) {\small{rightmost}};
\end{tikzpicture}
\caption{}
\label{subfig:q2_relevant}
\end{subfigure}
\begin{subfigure}{0.32\textwidth}
\centering
\begin{tikzpicture}[scale=0.8, every node/.style={scale=0.7}]
\renewcommand{\Aboundary}{3.5}
\nontermA
\renewcommand{\Stwolength}{3.5}
\renewcommand{\Ptwolength}{1}
\renewcommand{\qoneoffset}{0.3*\Sonelength}
\renewcommand{\qtwooffset}{0.53*\Stwolength}
\Stwo{0.4*\Aboundary,3*\rheight}
\Sone{0.55*\Aboundary,1.5*\rheight}
\end{tikzpicture}
\caption{}
\label{subfig:q1_included}
\end{subfigure}
\begin{subfigure}{0.32\textwidth}
\centering
\begin{tikzpicture}[scale=0.8, every node/.style={scale=0.7}]
\nontermA

\Stwo{0.66*\Aboundary,3*\rheight}
\Sone{0.3*\Aboundary,1.5*\rheight}
\end{tikzpicture}
\caption{}
\label{subfig:q1_overlap}
\end{subfigure}

%-----------------------
\vspace{0.5cm}

\setcounter{subfigure}{0}
\renewcommand\thesubfigure{2.\arabic{subfigure}}
\begin{subfigure}{0.45\textwidth}
\centering
\begin{tikzpicture}[scale=0.8, every node/.style={scale=0.7}]
\renewcommand{\Aboundary}{3.5}
\nontermA
\renewcommand{\Sonelength}{1.5}
\renewcommand{\Stwolength}{3.5}
\renewcommand{\Ptwolength}{1.5}
\renewcommand{\qoneoffset}{0.3*\Sonelength}
\renewcommand{\qtwooffset}{0.4*\Stwolength}
\Stwo{0.4*\Aboundary,3*\rheight}
\Sone{0.55*\Aboundary,1.5*\rheight}
\end{tikzpicture}
\caption{}
\label{subfig:q1_included_per}
\end{subfigure}
\begin{subfigure}{0.45\textwidth}
\centering
\begin{tikzpicture}[scale=0.8, every node/.style={scale=0.7}]
\renewcommand{\Aboundary}{3.5}
\nontermA
\renewcommand{\Stwolength}{2}
\renewcommand{\Ptwolength}{1}
\renewcommand{\qoneoffset}{0.3*\Sonelength}
\renewcommand{\qtwooffset}{0.35*\Stwolength}
\Stwo{0.75*\Aboundary,3*\rheight}
\Sone{0.55*\Aboundary,1.5*\rheight}
\end{tikzpicture}
\caption{}
\label{subfig:q1_overlap_per}
\end{subfigure}

%-----------------------
\caption{Assume that $S_1$ does not contain~$S_2$. The figure shows all possible locations of occurrences $p_1, p_2$ of $S_1,S_2$ in $\str{A}$. \textbf{In Case~\ref{it:aperiodic} of Corollary~\ref{cor:q1_and_q2}}, there are six subcases: \psubref{subfig:both_irrelevant} $p_1+|S_1|-1 \le |\str{\head(A)}|$, $p_2 > |\str{\head(A)}|$; \psubref{subfig:q1_relevant} $p_1$ is a relevant occurrence of $S_1$, $p_2 > |\str{\head(A)}|$; \psubref{subfig:both_relevant} $p_1,p_2$ are relevant; \psubref{subfig:q2_relevant} $p_2$ is relevant, $p_1+|S_1|-1 \le p_2$; \psubref{subfig:q1_included} $p_2$ is relevant, $p_2 < p_1 \le p_1+|S_1|-1 \le p_2 + |S_2|-1$; \psubref{subfig:q1_overlap} $p_2$ is relevant, $p_1 < p_2 < p_1+|S_1|-1 \le p_2+|S_2|-1$. By the definition of a co-occurrence and by Observation~\ref{obs:rev_q1_and_q2}, in Subcases~\psubref{subfig:both_irrelevant} and~\psubref{subfig:q2_relevant} $p_1$ must be as far to the right as possible, and in Subcases~\psubref{subfig:both_irrelevant} and~\psubref{subfig:q1_relevant} $p_2$ must be as far to the left as possible. \textbf{In Case~\ref{it:periodic}}, there are two subcases: \psubref{subfig:q1_included_per} $p_2$ is relevant and $p_2 \le p_1 \le p_1+|S_1|-1 \le p_2+|S_2|-1$; \psubref{subfig:q1_overlap_per} $p_2$ is relevant and $p_1 < p_2 < p_2+\pi_1-1 \le p_1+|S_1|-1$, where $\pi_1$ is the period of $S_1$. In all subcases, $q_2=p_2+\Delta_2$. In Subcases~\psubref{subfig:both_irrelevant}-\psubref{subfig:q1_overlap} $q_1=p_1+\Delta_1$ and in Subcases~\psubref{subfig:q1_included_per} and \psubref{subfig:q1_overlap_per} $q_1=p_1+\Delta_1+k\cdot \pi_1$ for some integer $k$.}
\label{fig:q1_and_q2}
\end{figure}
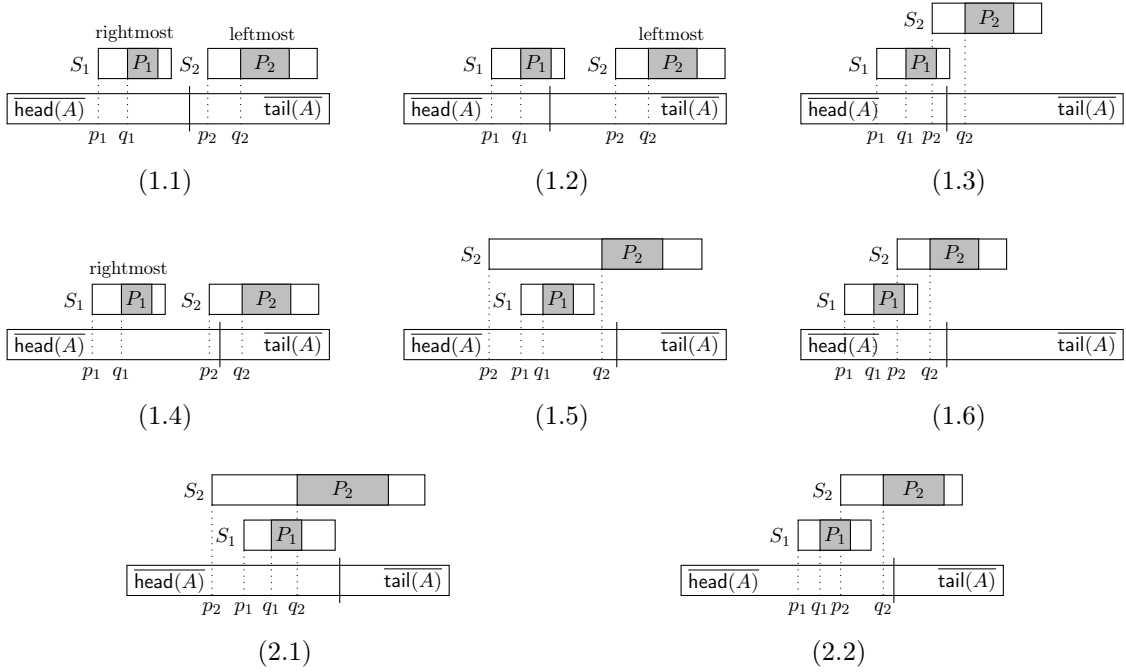  

\underline{Subcases~\psubref{subfig:both_irrelevant}-\psubref{subfig:q2_relevant}}. To retrieve the non-terminals, we query $\mathcal{T}_1(u_1, u_2,v_1, v_2)$ to find all integers that belong to the range $[\Delta,\Delta+b]$ (and the corresponding non-terminals). Recall that, for each non-terminal $A$, the tree stores an integer $d = p_2-p_1$, where $p_1$ is the starting position of an occurrence of $S_1$ in $\str{A}$ and $p_2$ of $S_2$. By Observation~\ref{obs:rev_q1_and_q2}, $p_1+\Delta_1$ is an occurrence of $P_1$ and $p_2+\Delta_2$ is an occurrence of $P_2$. The distance between them is in $[0,b]$ iff $d \in [\Delta,\Delta+b]$. By Observation~\ref{obs:close}, each retrieved non-terminal contains a close co-occurrence of $(q_1,q_2)$. On other other hand, if $\str{A}$ contains a co-occurrence $(q_1,q_2)$ corresponding to one Subcases~\psubref{subfig:both_irrelevant}-\psubref{subfig:q2_relevant}, then by Corollary~\ref{cor:q1_and_q2}, $p_1 = q_1-\Delta_1$ is an occurrence of $S_1$ and $p_2 = q_2-\Delta_2$ is an occurrence of $S_2$ and by construction $\mathcal{T}_1(u_1, u_2,v_1, v_2)$ stores an integer $d = p_2-p_1$. Therefore, the query retrieves all non-terminals corresponding to Subcases~\psubref{subfig:both_irrelevant}-\psubref{subfig:q2_relevant}. 

\underline{Subcases~\psubref{subfig:q1_included} and~\psubref{subfig:q1_included_per}}. We must decide whether an occurrence of $P_1$ in $S_2$ forms a $b$-close co-occurrence with the occurrence $\Delta_2$ of $P_2$ in $S_2$, and if so, report all non-terminals such that their expansion contains a relevant co-occurrence of $S_2$ with a split $l_2$, which are exactly the non-terminals stored in the list $\mathcal{L}(u_2, v_2)$. Let $k = \lceil \log(s_2) \rceil$. Recall that the index stores the following information for $k$: 
\begin{enumerate}
\item $p_1$, the rightmost occurrence of $S_1$ in $S_2$ such that $p_1+(|S_1|-1) \le l_2-2^k$;
\item $p'_1$, the leftmost occurrence of $S_1$ in $S_2$ such that $p_1' \le l_2-2^k \le p_1+(|S_1|-1)$;
\item $p''_1$, the rightmost occurrence of $S_1$ in $S_2$ such that $p_1'' \le l_2-2^k \le p''_1+(|S_1|-1)$.
\end{enumerate}
(See Fig.~\ref{fig:included}). By Observation~\ref{obs:rev_q1_and_q2}, the occurrence $p_1$ of $S_1$ induces an occurrence $q_1 = p_1+\Delta_1$ of $P_1$. Furthermore, if $S_1$ is periodic with period $\pi_1$, then $q_1+\pi_1\cdot k$, $0 \le k \le \lfloor (|S_1|-q_1-|P_1|)/\pi_1 \rfloor$, are also occurrences of $P_1$. One can decide whether the distance from any of these occurrences to $q_2$ is in $[0,b]$ in constant time, and if yes, then there $S_2$ contains a $b$-close co-occurrence of $P_1,P_2$ by Observation~\ref{obs:close}. Second, by Corollary~\ref{cor:arithmetic_progression}, if $S_1$ is not periodic, then there are no occurrences of $S_1$ between $p_1'$ and $p_1''$ and $p_1',p_1''$ by Observation~\ref{obs:rev_q1_and_q2} induce occurrences $p_1'+\Delta_1, p_1''+\Delta_1$ of $P_1$. Otherwise, there are occurrences of $P_1$ in every position $p_1'+\Delta_1+k\cdot\pi_1$, $0 \le k \le \lfloor (|S_1|+p_1''-|P_1|-p_1')/\pi_1 \rfloor$. Similarly, we can decide whether the distance from any of them to the occurrence $\Delta_2$ of $P_2$ in $S_2$ is in $[0,b]$ in constant time. Finally, let $q_1$ be the rightmost occurrence of $P_1$ in $S_2$ in the interval $[l_2-2^k+1, \Delta_2]$. We extract $S_2(l_2-2^k, \Delta_2 +|P_2|)$ via Fact~\ref{fact:prefsuf_extraction} and search for $q_1$ using a linear-time pattern matching algorithm for $P_1$, which takes $O(|P_1|+|P_2|) = O(m)$ time. If $0 \le \Delta_2-q_1\le b$, then there is a $b$-close co-occurrence of $P_1,P_2$ in $S_2$. Correctness follows from Corollary~\ref{cor:q1_and_q2}, Observation~\ref{obs:rev_q1_and_q2} and Observation~\ref{obs:close}. 

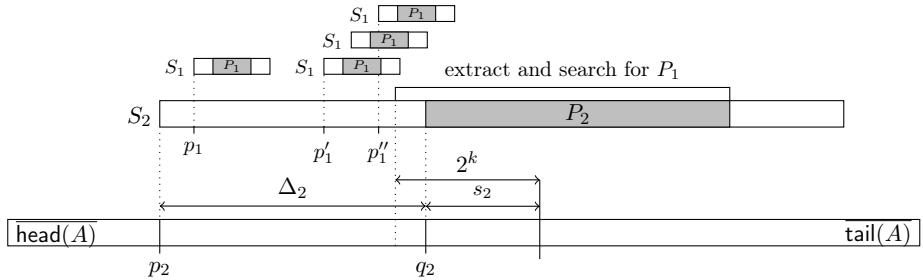
\begin{figure}[h!]
\centering
\begin{tikzpicture}[scale=1, every node/.style={scale=0.8}]
%------------------------------------------------------------ COMMANDS
\newcommand{\rheight}{0.35}
\newcommand{\Sone}[1]{
\draw (#1) node[left,shift={(0,{0.3*\rheight})}]  {\footnotesize{$S_1$}} rectangle ($(#1)+(\Sonelength,0.6*\rheight)$);
\draw[fill=gray!50] ($(#1)+(0.25*\Sonelength,0)$) rectangle ++ (\Ponelength,0.6*\rheight) node[midway]  {\tiny{$P_1$}};
}

\newcommand{\Sonelength}{1}
\newcommand{\Ponelength}{0.5}
\newcommand{\Stwolength}{9}
\newcommand{\Ptwolength}{4}
\newcommand{\Ptwostart}{3.5}
\newcommand{\Aboundary}{5}

%----------------------------------------------------------------------

%Non terminal A
\begin{scope}[shift={(0,{-4.5*\rheight})}]
	\node (ha) at (-2,0) {};
	\draw (ha) node[right, shift={(0,{.55*\rheight})}] {$\str{\head(A)}$} rectangle ++ (12,\rheight) node[left, shift={(0,{-.65*\rheight})}] (ta) {$\str{\tail(A)}$};

	% boundary	
	\draw (\Aboundary,-0.5*\rheight) -- (\Aboundary,2.5*\rheight);
	
	% annotation
	\node[above] (p2) at (0,\rheight) {};
	\draw (p2) -- ($(p2) - (0,1.6*\rheight)$) node[below] {$p_2$};
	\draw[dotted] ($(p2) - (0,0.2*\rheight)$) -- ($(p2) + (0,3*\rheight)$);
	
	\node[above] (q2) at (\Ptwostart,\rheight) {};
	\draw (q2) -- ($(q2) - (0,1.6*\rheight)$) node[below] {$q_2$};
	\draw[dotted] ($(q2) - (0,0.2*\rheight)$) -- ($(q2) + (0,3*\rheight)$);
	\draw[<->] ($(p2)+(0,0.2*\rheight)$) -- ($(q2)+(0,0.2*\rheight)$) node[midway,above] {$\Delta_2$};
	\draw[<->] (\Ptwostart,1.5*\rheight) -- (\Aboundary,1.5*\rheight) node[midway,above] {\small{$s_2$}};
	
	%2^k
	 \node (twok) at (3.1,0) {};
	 \draw[<->] ($(twok)+(0,2.5*\rheight)$) -- (\Aboundary,2.5*\rheight) node[midway,above] {$2^k$};
	 \draw[dotted] ($(twok) - (0,0)$) -- ($(twok) + (0,5.5*\rheight)$); 
\end{scope}

%S2
\node[left] (s2) at (0,0.5*\rheight) {};
%\node[left] () at (0,2*\rheight) {$S_1$};
\draw (0,0) node[left,,shift={(0,{0.5*\rheight})}] {$S_2$} rectangle ++ (\Stwolength,\rheight) ;
\draw[fill=gray!50] (\Ptwostart,0) rectangle ++ (\Ptwolength,\rheight) node (endp2) {} node[midway]  {$P_2$};

% extract and search for P_1
\draw ($(twok) + (0,5.5*\rheight)$) -- ($(twok)+(0,6*\rheight)$) -- ($(endp2)+(0,0.5*\rheight)$) node[midway, above] {\small{extract and search for $P_1$}} -- ($(endp2)-(0,0)$) ;

%p''1
\node (p1) at (0.32*\Stwolength,0) {};
\draw ($(p1)+(0,0.01)$) -- ($(p1)+(0,-0.25*\rheight)$) node[below] {\small{$p''_1$}};
\draw[dotted] ($(p1)+(0,0)$) -- ($(p1)+(0,4*\rheight)$);
\Sone{0.32*\Stwolength,4*\rheight};

\Sone{0.28*\Stwolength,3*\rheight};

%p'1
\node (p1) at (0.24*\Stwolength,0) {};
\draw ($(p1)+(0,0.01)$) -- ($(p1)+(0,-0.25*\rheight)$) node[below] {\small{$p'_1$}};
\draw[dotted] ($(p1)+(0,0)$) -- ($(p1)+(0,2*\rheight)$);
\Sone{0.24*\Stwolength,2*\rheight};

%p1
\node (p1) at (0.05*\Stwolength,0) {};
\draw ($(p1)+(0,0.01)$) -- ($(p1)+(0,-0.25*\rheight)$) node[below] {\small{$p_1$}};
\draw[dotted] ($(p1)+(0,0)$) -- ($(p1)+(0,2*\rheight)$);
\Sone{0.05*\Stwolength,2*\rheight};

\end{tikzpicture}
\caption{Query algorithm for Subcases~\psubref{subfig:q1_included} and~\psubref{subfig:q1_included_per}.}
\label{fig:included}
\end{figure}

{\underline{Subcase~\psubref{subfig:q1_overlap_per}}}. Let $\pi_1$ be the period of $S_1$. We retrieve the non-terminals associated with the integers $q \in \mathcal{T}_2(u_1, u_2,v_1, v_2)$ such that the intersection of an interval $I = [a,b]$ and $[\ell,q]$ is non-empty, where 
$$a=\lceil{(\Delta-\ov(S_1,S_2))/\pi_1\rceil} \text{, } b=\lfloor{(\Delta-\ov(S_1,S_2)+b)/ \pi_1 \rfloor} \text{ and } \ell = -\lfloor{(|S_1|-|P_1|-\Delta_1)/\pi_1\rfloor}$$
(See the description of the index for the definition of $\ov(S_1,S_2)$). As $\ell$ is fixed, we can implement the query via at most one binary tree search: If $b \le \ell$, the output is empty, if $a \le \ell \le b$, we must output all integers, and if $\ell \le a$, we must output all $q \ge b$. Let us now explain why the algorithm is correct. Consider a non-terminal $A$ for which $\mathcal{T}_2(u_1, u_2,v_1, v_2)$ stores an integer $q$. By construction, $\str{A}$ contains a relevant occurrence of $S_2$ with a split $l_2$. A position $p_1 = |\str{\head(A)}|-l_2-\ov(S_1,S_2)-q \cdot \pi_1$ is the leftmost occurrence of $S_1$ in $\str{A}$ such that $p_1 \le p_2 \le p_1+|S_1|-1$ and $p_2 = |\str{\head(A)}|-l_2-\ov(S_1,S_2)$ the rightmost. Consequently, there is an occurrence $q_1 = |\str{\head(A)}|-l_2-\ov(S_1,S_2)-q'\cdot \pi_1+\Delta_1$ of $P_1$ for each $-\lfloor{(|S_1|-|P_1|-\Delta_1)/\pi_1\rfloor} \le q' \le q$. The occurrence of $S_2$ implies that $q_2 = |\str{\head(A)}|-s_2$ is an occurrence of $P_2$. We have $0 \le q_2-q_1=q'\cdot\pi_1+\ov(S_1,S_2)-\Delta\le b$ iff $\Delta-\ov(S_1,S_2) \le  q' \cdot \pi_1 \le \Delta-\ov(S_1,S_2)+b$, which is equivalent to $[\ell,q] \cap I \neq \emptyset$. It follows that we retrieve every non-terminal corresponding to Subcase~\psubref{subfig:q1_overlap_per}. On the other hand, by Observation~\ref{obs:close}, the expansion of each retrieved non-terminal contains a $b$-close co-occurrence of $P_1, P_2$.

{\underline{Subcase~\psubref{subfig:q1_overlap}}}. We argue that we have already reported all non-terminals corresponding to this subcase and there is nothing left to do. Consider a non-terminal $A$ such that its expansion contains a relevant occurrence $p_2$ of $S_2$. If there are at most two occurrences $p_1$ of $S_1$ such that $p_1 \le p_2 \le p_1+|S_1|-1\le p_2+|S_2|-1$, we will treat them when we query $\mathcal{T}_1(u_1, u_2,v_1, v_2)$ (Subcases~\psubref{subfig:both_irrelevant}-\psubref{subfig:q2_relevant}). Otherwise, by Corollary~\ref{cor:arithmetic_progression}, $S_1$ is periodic and there is an occurrence $p'_1$ of $S_1$ such that $p'_1 \le p_2 < p_2 + \pi_1 \le p_1+|S_1|-1 < p_2 + |S_2|-1$. The non-terminals corresponding to this case are reported when we query $\mathcal{T}_2(u_1, u_2,v_1, v_2)$ (Subcase~\psubref{subfig:q1_overlap_per}).

{\underline{Time complexity}}. 
As shown above, the algorithm reports a set $\mathcal{N}' \supset \mathcal{N}$ of non-terminals and each non-terminal in $\mathcal{N}'$ contains a $b$-close co-occurrence. By Claim~\ref{claim:relevant_cons_occ} and since the height of $G'$ is $h = O(\log N)$, we have $|\mathcal{N}'| = O(\occ \log N)$. Furthermore, for a fixed pair of splits of $P_1,P_2$, each non-terminal in $\mathcal{N}'$ can be reported a constant number of times. Since $|\splits(G',P_1)| \cdot |\splits(G',P_2)| = O(\log^2 N)$, the total size of the output is $|\mathcal{N}'| \cdot O(\log^2 N) = O(\occ \cdot \log^3 N)$. We therefore obtain that the running time of the algorithm is $O(m+\log^3 N + \occ \log^3 N) = O(m+(1+\occ)\log^3 N)$ as desired. 
\end{proof}

Once we have retrieved the set $\mathcal{N}'$, we find all $b$-close relevant co-occurrences for each of the non-terminals in $\mathcal{N}'$ using Theorem~\ref{th:occurrences}. In fact, our algorithm acts naively and computes \emph{all} relevant co-occurrences for a non-terminal in $\mathcal{N}'$, and then selects those that are $b$-close. By case inspection, one can show that a relevant co-occurrence for a non-terminal $A$ always consists of an occurrence of $P_2$ that is either relevant or the leftmost in $\str{\tail(A)}$, and a preceding occurrence of $P_1$. Intuitively, this allows to compute all relevant co-occurrences efficiently and guarantees that their number is small. Formally, we show the following claim: 

\begin{restatable}{lemma}{computingrelevant}\label{lem:relevant_co_occurr_A}
Assume that $P_2$ is not a substring of $P_1$. After $O(m \log N + \log^2 N)$-time preprocessing, the data structure of Theorem~\ref{th:occurrences} allows to compute all $b$-close relevant co-occurrences of $P_1, P_2$ in the expansion of a given non-terminal $A$ in time $O(\log^{3} N \log\log N)$. 
\end{restatable}

A part of the index of Christiansen et al.~\cite{talg/ChristiansenEKN21} is a pruned copy of the parse tree of $G'$. They showed how to traverse the tree to report all occurrences of a pattern, given its relevant occurrences in the non-terminals. By using essentially the same algorithm, we can report all $b$-close co-occurrences in amortized constant time per co-occurrence, which concludes the proof of Theorem~\ref{thm:close_co_occurrences}. (For completeness, we provide all details of this step in Appendix~\ref{app:close}, Lemma~\ref{lem:close_co_occurr}.)

\newpage

\bibliographystyle{plainurl}
\bibliography{bibliography}
\newpage
\appendix
\section{Proofs omitted from Section~\ref{sec:prelim}}\label{app:proofs}

\tries*
\begin{proof}
Let us first recall the definition of the Karp--Rabin fingerprint.

\begin{definition}[Karp--Rabin fingerprint]
For a prime~$p$ and an~$r \in \mathbb{F}_p^\ast$, the \emph{Karp--Rabin fingerprint}~\cite{karp1987efficient} of a string~$X$ is defined as a tuple $(r^{|X|-1} \mod p, r^{-|X|+1} \mod p, \varphi_{p,r}(X))$, where $\varphi_{p,r}(X)=\sum_{k=0}^{|X|-1} S[k]r^k\mod p$.
\end{definition}

We use the following result of Christiansen et al.~\cite{talg/ChristiansenEKN21}, which builds on Belazzougui~et~al.~\cite{esa/BelazzouguiBPV10} and Gagie~et~al.~\cite{latin/GagieGKNP14,soda/GagieNP18}.
 
\begin{fact}[{\cite[Lemma 6.5]{talg/ChristiansenEKN21}}]\label{fact:compact_trie}
Let $\mathcal{S}$ be a set of strings and assume we have a data structure supporting extraction of any length-$l$ prefix of strings in $\mathcal{S}$ in
time $f_e(l)$ and computing the Karp--Rabin fingerprint $\varphi$ of any length-$l$ prefix of a string in $\mathcal{S}$
in time $f_h(l)$. We can then build a data structure that uses $O(|\mathcal{S}|)$ space and supports the following queries in $O(m + f_e (m) + \tau ( f_h (m) + \log m))$ time: Given a pattern $P$ of length $m$ and $\tau > 0$
suffixes $Q_1,\dots,Q_{\tau}$ of $P$, find the intervals of strings in (the lexicographically-sorted) $\mathcal{S}$ prefixed by
$Q_1,\dots,Q_{\tau}$.
\end{fact}

It should be noted that despite using a hash function, the query algorithm is deterministic: the proof shows that $p$ and $r$ can be chosen during the construction time to ensure that there are no collisions on the substrings of the strings in $\mathcal{S}$.  

To bound $f_e$, we use~\cite[Lemma 6.6]{talg/ChristiansenEKN21} which builds on G\k{a}sieniec et al.~\cite{dcc/GasieniecKPS05} and Claude and Navarro~\cite{spire/ClaudeN12a}.

\begin{fact}[{\cite[Lemma 6.6]{talg/ChristiansenEKN21}}]\label{fact:prefsuf_extraction} Given an RLSLP
of size $O(g)$, there exists a data structure of size $O(g)$ such that any length-$l$ prefix or suffix of $\str{A}$ can be
obtained from any non-terminal $A$ in time $f_e(l) = O(l)$.
\end{fact}

To bound $f_h(l)$, we introduce a simple construction based on the following well-known fact:

\begin{fact}\label{fact:fingerprints}
Consider strings $X,Y,Z$ where $XY = Z$. Given the Karp--Rabin fingerprints of two of the three strings, one can compute the fingerprint of the third string in constant time.
\end{fact}

\begin{claim}\label{claim:fingerprint_extraction}
Given a RLSLP $G$ of size $g$ and height $h$, there exists a data structure of size $O(g)$ that given a non-terminal $A$ and an integer $l$ allows to retrieve the Karp-Rabin fingerprints of the length-$l$ prefix and suffix of $\str{A^r}$ and $\rev{\str{A^r}}$ in time $f_h(l) = O(h + \log l)$.
\end{claim}
\begin{proof}
The claim for $\rev{\str{A^r}}$ follows for the claim for $\str{A^r}$ by considering the grammar $G_{rev}$, where the order of the non-terminals in each production is reversed. Below we focus on extracting the fingerprints for $\str{A^r}$, and we further restrict our attention to prefixes of $\str{A^r}$, the algorithm for suffixes being analogous. 

The data structure consists of two sets. The first set contains the lengths of the expansions of all non-terminals in the grammar, and the second one their fingerprints. 
 
By Fact~\ref{fact:fingerprints} and doubling, it suffices to show an algorithm for computing the fingerprint of the length-$l$ prefix of $\str{A}$. Assume that $A$ associated with a rule $A\rightarrow BC$. If the length of $\str{A}$ is smaller than $l$, we return error. Otherwise, to compute the fingerprint of the length-$l$ prefix of $\str{A}$, we consider two cases. If $l\leq |\str{B}|$, we recurse on $B$ to retrieve the fingerprint of the $l$-length prefix of $\str{B}$. Otherwise, we recurse on $C$ to retrieve the fingerprint of $\str{C}[\dots l-|\str{B}|)$ and then compute the fingerprint of the $l$-length prefix of $\str{A}$ from the fingerprints of $\str{B}$ and $\str{C}[\dots l-|\str{B}|)$ in constant time by Fact~\ref{fact:fingerprints}. 

For a non-terminal $A$ associated with a rule $A\rightarrow B^{r}$, we compute the fingerprint analogously. If the length of $\str{A}$ is smaller than $l$, we return error. Otherwise, let $q$ be such that $q \cdot |\str{B}| \leq l < (q+1) \cdot |\str{B}|$. We compute the fingerprint of $\str{B}^q$ from the fingerprint of $\str{B}$ by applying Fact~\ref{fact:fingerprints} $O(1+\log q)$ times, and the fingerprint of $\str{B}[\dots l-q \cdot |\str{B}|)$ recursively. We can then apply Fact~\ref{fact:fingerprints} to compute the fingerprint of the length-$l$ prefix of $\str{A}$ in constant time. Note that in this case, the length of the prefix decreases by a factor at least $q$. 

If we are in a terminal $A$, the calculation takes $O(1)$ time (the prefix must be equal to $A$ itself). 
 
In total, we spend $O(h+\log l)$ time as we recurse $O(h)$ times, and whenever we spend more than constant time in a symbol, we charge it on the decrease in the length. The fingerprints of length-$l$ suffixes are computed analogously.
\end{proof}

By substituting the bounds for $f_e(l)$ (Fact~\ref{fact:fingerprints}) and $f_h(l)$ (Claim~\ref{fact:prefsuf_extraction}) into Fact~\ref{fact:compact_trie}, we obtain the claim of the lemma.
\end{proof}

\section{Proofs omitted from Section~\ref{sec:occurrences}}
\label{app:occurrences}
\begin{claim}\label{claim:leftmost_rightmost}
Given a non-terminal $A$ of $G'$, we can find the leftmost and the rightmost occurrences of $P$ in $\str{A}$ and as a corollary in $\str{\head(A)}$ and $\str{\tail(A)}$ in $O(\log^{2} N\log\log N)$ time.  % h splits \log^{\eps} N
\end{claim}
\begin{proof}
We explain how to find the leftmost occurrence of $P$ in $\str{A}$, the rightmost one can be found analogously. We first check whether $\str{A}$ contains an occurrence of $P$ via Claim~\ref{claim:emptiness} in $O(\log N\log \log N)$ time. If it does not, we can stop immediately. Below we assume that there is an occurrence of $P$ in  $\str{A}$. Next, we check whether $\str{\head(A)}$ contains an occurrence of $P$ via Claim~\ref{claim:emptiness} in $O(\log N\log \log N)$ time. If it does, the leftmost occurrence of $P$ in $\str{A}$ is the leftmost occurrence of $P$ in $\str{\head(A)}$ and we can find it by recursing on $\head(A)$. If $\str{\head(A)}$ does not contain an occurrence of $P$, but $\str{A}$ contains relevant occurrences of $P$, then the leftmost occurrence of $P$ in $\str{A}$ is the leftmost relevant occurrence of $P$ in $\str{A}$ and we can find it in $O(|\splits(G',P)|) = O(\log N)$ time. Finally, if $P$ neither occurs in $\str{\head(A)}$ nor has relevant occurrences in $\str{A}$, then the leftmost occurrence of $P$ in $\str{A}$ is the leftmost occurrence of $P$ in $\str{\tail(A)}$. If $\tail(A)$ is a non-terminal $C$, we recurse on $C$ to find it. If $\tail(A)=B^{r-1}$ for a non-terminal $B$, $\str{\tail(A)}$ cannot contain an occurrence of $P$ because $\str{B}$ does not contain $P$ and there are no relevant occurrences in $A$. We recurse down at most $h = O(\log N)$ levels, and spend $O(\log N\log \log N)$ time per level. The claim follows.
\end{proof}

\begin{lemma}\label{lm:predecessor}
Let $A$ be a non-terminal of $G'$. For any position $p$, we can find the rightmost occurrence $q \le p$ of $P$ in $\str{A}$ and the leftmost occurrence $q'\geq p$ of $P$ in $\str{A}$ in $O(\log^{3} N\log \log N)$ time.  
\end{lemma}
\begin{proof}
First we describe how to locate $q$. Consider a node $u$ of the parse tree of $G'$ labeled by $A$. The algorithm starts at $u$ and recurses down. Let $A'$ be the label of the current node. It computes the leftmost and rightmost occurrences in $\str{A'}$, $\str{\head(A')}$ and $\str{\tail(A')}$ as well as all relevant occurrences via Claim \ref{claim:leftmost_rightmost}. If the leftmost occurrence of $P$ in $\str{A'}$ is larger than $p$, the search result is empty. Otherwise, consider two cases. 

\begin{enumerate}
\item $A'$ is associated with a rule $A' \rightarrow B'C'$, i.e. $\head(A') = B'$, $\tail(A') = C'$. 
\begin{enumerate}
\item If $p \le |\str{B'}|$, recurse on $B'$. 
\item Assume now that $p > |\str{B'}|$. If the leftmost  occurrence of $P$ in $\str{C'}$ is smaller than $p$, recurse on $C'$. Otherwise, return the rightmost relevant occurrence of $P$ in $\str{A'}$ if it exists else the rightmost occurrence of $P$ in $\str{B'}$. 
\end{enumerate}
\item $A'$ is associated with a rule $A \rightarrow (B')^r$, i.e. $\head(A') = B'$, $\tail(A') = (B')^{r-1}$.  Let an integer $k$ be such that $(k-1) \cdot |\str{B'}|+1 \le p \le k \cdot |\str{B'}|$. The desired occurrence of $P$ is the rightmost one of the following ones:
\begin{enumerate}
	\item The rightmost occurrence $q \le p$ of $P$ which crosses the border between two copies of $\str{B'}$. To compute $q$, we compute all relevant occurrences of $P$ in $\str{A'}$ and then shift each of them by the maximal possible shift $r' \cdot |\str{B'}|$, where $r'$ is an integer, which guarantees that it starts before $p$ and ends before $|\str{A'}|$ and take the rightmost of the computed occurrences to obtain $q$.
	\item The rightmost occurrence $q$ of $P$ such that for some integer $k'$, we have $(k'-1) \cdot |\str{B'}| \le q \le q+|P|-1 \le k' \cdot |\str{B'}|$ (i.e. the occurrence fully belongs to some copy of $\str{B'}$). In this case, $q$ is either the rightmost occurrence of $P$ in the $(k-1)$-th copy of $\str{B'}$, or the rightmost occurrence of $P$ in the $k$-th copy of $\str{B'}$ that is smaller than $p$. In the second case, we compute $q$ by recursing on $B'$.
	\end{enumerate}
\end{enumerate}
 We recurse down at most $h$ levels. On each level we spend $O(\log^{2} N\log\log N)$ time to compute the leftmost, the rightmost, and relevant occurrences and respective shifts for a constant number of non-terminals via Claim~\ref{claim:leftmost_rightmost}. Therefore, in total we spend $O(h \cdot \log^{2} N\log\log N) = O(\log^{3} N\log\log N)$ time.

Locating $q'$ is very similar and differs only in small technicalities. The algorithm starts at the node $u$ and recurses down. Let $A'$ be the label of the current node. We compute the leftmost and rightmost occurrences in $\str{A'}$, $\str{\head(A')}$ and $\str{\tail(A')}$ as well as all relevant occurrences via Claim \ref{claim:leftmost_rightmost}. If the rightmost occurrence of $P$ in $\str{A'}$ is smaller than $p$, the search result is empty. Otherwise, consider two cases. 

\begin{enumerate}
\item $A'$ is associated with a rule $A' \rightarrow B'C'$, i.e. $\head(A') = B'$, $\tail(A') = C'$. 
\begin{enumerate}
\item If $p > |\str{B'}|$, recurse on $C'$. 
\item Assume now that $p \leq |\str{B'}|$. If the rightmost  occurrence of $P$ in $\str{B'}$ is larger than $p$, recurse on $B'$. Otherwise, return the leftmost relevant occurrence $q$ satisfying $q\geq p$, if it exists, and otherwise the leftmost occurrence of $P$ in $\str{C'}$. %Otherwise, return the rightmost relevant occurrence of $P$ in $\str{A'}$ if it exists else the rightmost occurrence of $P$ in $\str{B'}$. 
\end{enumerate}
\item $A'$ is associated with a rule $A \rightarrow (B')^r$, i.e. $\head(A') = B'$, $\tail(A') = (B')^{r-1}$.  Let an integer $k$ be such that $(k-1) \cdot |\str{B'}|+1 \le p \le k \cdot |\str{B'}|$. The desired occurrence of $P$ is the leftmost one of the following ones:
\begin{enumerate}
	\item The leftmost occurrence $q' \geq p$ of $P$ which crosses the border between two copies of $\str{B'}$. To compute $q'$, we compute all relevant occurrences of $P$ in $\str{A'}$ and then shift each of them by the minimal possible shift $r' \cdot |\str{B'}|$, where $r'$ is an integer, which guarantees that it starts after $p$ and ends before $|\str{A'}|$ (if it exists) and take the leftmost of the computed occurrences to obtain $q$.
	\item The leftmost occurrence $q'$ of $P$ such that for some integer $k'$, we have $(k'-1) \cdot |\str{B'}| \le q' \le q'+|P|-1 \le k' \cdot |\str{B'}|$ (i.e. the occurrence fully belongs to some copy of $\str{B'}$). In this case, $q'$ is either the leftmost occurrence of $P$ in the $(k+1)$-st copy of $\str{B'}$, or the leftmost occurrence of $P$ in the $k$-th copy of $\str{B'}$ that is larger than $p$. In the second case, we compute $q'$ by recursing on $B'$.
	\end{enumerate}
\end{enumerate}
The time complexities are the same as for computing $q$.
\end{proof}

\section{Proofs omitted from Section~\ref{sec:close}}
\label{app:close}

\computingrelevant*
\begin{proof}
We preprocess $P_1, P_2$ in $O(m \log N + \log^2 N)$ time as explained in Theorem~\ref{th:occurrences}. Upon receiving a non-terminal $A$, we compute the leftmost and the rightmost occurrences of $P_1, P_2$ in $\str{\head(A)}$ and $\str{\tail(A)}$, as well as a set $\Pi_1$ of all relevant occurrences of $P_1$ in $\str{A}$ and a set $\Pi_2$ of all relevant occurrences of $P_2$ in $\str{A}$ via Claim~\ref{claim:leftmost_rightmost}.
We will compute all relevant co-occurrences in $\str{A}$, selecting those of them that are $b$-close is then trivial. 
As $q_1 \le q_2$ by definition, each relevant co-occurrence $(q_1,q_2)$ of $P_1,P_2$ in $\str{A}$ falls under one of the following categories:

\begin{enumerate}
\item $q_1$ is a relevant occurrence of $P_1$ in $\str{A}$ and $q_2$ is a relevant occurrence of $P_2$ in $\str{A}$ (i.e. $q_1 \in \Pi_1, q_2 \in \Pi_2$). To check whether a pair $q_1 \in \Pi_1, q_2 \in \Pi_2$ forms a co-occurrence of $P_1, P_2$ in $\str{A}$, we must check whether there is an occurrence $q$ of either $P_1$ or $P_2$ between $q_1$ and $q_2$. The occurrence $q$ can only be the rightmost occurrence $r_q$ of $P_2$ in $\str{\head(A)}$, the leftmost occurrence $l_q$ of $P_1$ in $\str{\tail(A)}$, or an occurrence in $\Pi_1 \cup \Pi_2$. Consequently, we can find all co-occurrences in this category by merging two (sorted) sets: $\Pi_1 \cup \{l_q\}$ and $\{r_q\} \cup \Pi_2$, which can be done in $O(2 + |\Pi_1 \cup \Pi_2|)$ time.

\item $1 \le q_1  \le q_1+|P_1|-1 \le |\str{\head(A)}|$ and  $|\str{\head(A)}| < q_2 \le q_2+|P_2|-1$. In this case, $q_1$ must be the rightmost occurrence of $P_1$ in $\str{\head(A)}$ and $q_2$ the leftmost occurrence in $\str{\tail(A)}$, $q_1 \le q_2$, and there must be no occurrence $q \in \Pi_1 \cup \Pi_2$  such that $q_1 \le q \le q_2$. Therefore, if there is a co-occurrence in this category, we can retrieve it in $O(|\Pi_1 \cup \Pi_2|)$ time.

\item $q_1$ is a relevant occurrence of $P_1$ in $\str{A}$ (i.e. $q_1 \in \Pi_1$) and $|\str{\head(A)}| < q_2 \le q_2+|P_2|-1$. In this case, $q_1$ must be the rightmost occurrence in $\Pi_1$ and $q_2$ the leftmost occurrence of $P_2$ in $\str{\tail(A)}$, and there should be no occurrence from $\Pi_2$ between $q_1$ and $q_2$. Therefore, if there is a co-occurrence in this category, we can find it in $O(|\Pi_1 \cup \Pi_2|)$ time.

\item $q_1 \le q_1+|P_1|-1 \le |\str{\head(A)}|$ and $q_2$ is a relevant occurrence of $P_2$ in $\str{A}$ (i.e. $q_2 \in \Pi_2$). First, consider the leftmost occurrence in $q_2 \in \Pi_2$. We find the rightmost occurrence $q_1 \le q_2$  of $P_1$ in $\str{A}$ via a predecessor query. The pair $(q_1,q_2)$  is a co-occurrence iff the rightmost occurrence of $P_2$ in $\str{\head(A)}$ is smaller than $q_1$, which can be checked in constant time. Second, we consider the remaining occurrences in $\Pi_2$. Let $q_2'$ be the leftmost one. We begin by computing the preceding occurrence $q_1'$ of $P_1$ via a predecessor query and if $q_2 \le q_1'$, output the resulting co-occurrence. If $\Pi_2 = \{q_2, q_2'\}$, we are done. Otherwise, by Corollary~\ref{cor:arithmetic_progression}, the occurrences in $\Pi_2 \setminus \{q_2\}$ form an arithmetic progression with difference equal to the period of $P_2$ (as all of them contain the position $|\str{\head(A)}|$). Furthermore, as $P_1$ does not contain $P_2$, the occurrence of $P_1$ preceding $q_2'$ belongs to the periodic region formed by the relevant occurrences of $P_2$. Therefore, all the remaining co-occurrences can be obtained from the co-occurrence for $q_2'$ by shifting them by the period. In total, this step takes $O(|\Pi_2| + \log^{3} N\log\log N)$ time.
\end{enumerate}
\end{proof}

\begin{lemma}
\label{lem:close_co_occurr}
Assume that $P_2$ is not a substring of $P_1$. One can compute all $b$-close co-occurrences of $P_1, P_2$  in $S$ in time $O(m + (1+\occ) \cdot \log^{4} N \log\log N)$. 
\end{lemma}
\begin{proof}
During the preprocessing, we prune the parse tree: First, for each non-terminal $B$, all but the first node labeled by $B$ in the preorder is converted into a leaf and its subtree is pruned. For each node $v$ labeled by a non-terminal $B$, we store $\anc(v)$, the nearest ancestor $u$ of $v$ labeled by $A$ such that $u$ is the root or $A$ labels more than one node in the tree. Second, for every node labeled by a non-terminal $A$ associated with a rule $A \rightarrow B^k$, we replace its $k-1$ rightmost children with a leaf labeled by $B^{k-1}$. We call the resulting tree \emph{the pruned parse tree} and for each node $v$ labeled by a non-terminal $B$ store $\nextnode(v)$, the next node labeled by $B$ in preorder, if there is one. As every non-terminal labels at most one internal node of the pruned parse tree and every node has at most two children, it occupies $O(g')$ space.

When the algorithm of Lemma~\ref{lm:non-term_close_co_occ} outputs $A \in \mathcal{N}'$, we compute all relevant co-occurrences $(q_1,q_2)$ in $\str{A}$ in time $O(\log^{3} N\log \log N)$ using Lemma~\ref{lem:relevant_co_occurr_A} select those which satisfy $q_2-q_1 \leq b$.

Fix a $b$-close relevant co-occurrence $(q_1,q_2)$ in $\str{A}$. If $A$ is associated with a rule $A \rightarrow BC$, construct a set $\occ(A) := \{(q_1,q_2)\}$, and otherwise if $A$ is associated with a rule $A \rightarrow B^k$,
$$\occ(A) := \{(q_1+i \cdot |\str{B}|, q_2+i \cdot |\str{B}|) : 0 \le i \le \lfloor (|\str{A}|-q_2-|P_2|+1)/|\str{B}| \rfloor\}$$
Suppose that $A$ labels nodes $v_1, v_2, \ldots, v_k$ of the unpruned parse tree of $G'$. If $W$ is a set of co-occurrences, denote for brevity $W+\delta = \{(q_1+\delta,q_2+\delta) : (q_1,q_2) \in W\}$. Below we show an algorithm that generates a set $\mathcal{S} = \cup_i \occ(A) + \off(v_i)$ that contains all secondary $b$-close co-occurrences due to $(q_1,q_2)$.  

We traverse the pruned parse tree, while maintaining a priority queue. The queue is initialized to contain the first node in the preorder labeled by $A$ together with $\occ(A)$. Until the priority queue is empty, pop a node $v$ and a set $W$ of co-occurrences of $P_1,P_2$ in the expansion of its label, and perform the following steps:
\begin{itemize}
\item\label{step:reporting} \underline{Reporting step:}  If $v$ is the root, report $W$;
\item \label{step:next} \underline{Next node step:} If $\nextnode(v)$ is defined, push $(\nextnode(v),W+\off(\nextnode(v))-\off(v))$;
\item \label{step:siblings} \underline {Sibling step:}  If $v$ is labeled by a non-terminal $B$ and its sibling by $B^k$, for some integer $k$, then $W := \cup_{0 \le i \le k} W+i \cdot |\str{B}|$
\item \label{step:anc} \underline{Duplicate ancestor step:} Push to the queue $(\anc(v),W+\off(\anc(v))-\off(v))$. 
\end{itemize}

By construction and as every node is connected with the root by a path of $\anc$ links, the algorithm generates all co-occurrences in $\mathcal{S}$. Let us show that it reports every co-occurrence in $\mathcal{S}$ at most once. Assume that a co-occurrence $(q_1',q_2') \in \mathcal{S}$ was reported twice by two different sequences of steps. Let $u_1,u_2$ be the nodes which were added to the queue after the last next node step in the sequences. We claim that either $u_1$ is an ancestor of $u_2$, or vice versa (otherwise, the subtrees of $u_1$ and $u_2$ are disjoint, and hence the sets of co-occurrences that we created for $u_1,u_2$ have an empty intersection). Assume w.l.o.g. that $u_1$ is an ancestor of $u_2$. However, as we arrived at $u_1$ by a $\nextnode$ link, then by definition the subtree of $u_1$ must be pruned, a contradiction.  

The time complexity follows: The algorithm of Lemma~\ref{lm:non-term_close_co_occ} takes $O(m +(1+\occ) \cdot \log^3 N)$ time; applying Lemma~\ref{lem:relevant_co_occurr_A} to every non-terminal in $\mathcal{N}'$ $O(\occ \cdot \log^{4} N\log\log N))$; and maintaining the queue and reporting the co-occurrences takes $O(\occ)$ time as at every step we can charge the time needed to update the queue on newly created occurrences, and each occurrence is reported ($\anc(\cdot)$ is always defined). 
\end{proof}

\end{document}